\documentclass[12pt]{article}
\usepackage{amsmath}
\usepackage{graphicx}
\usepackage{enumerate}
\usepackage{natbib}
\usepackage{url} 

\newcommand{\blind}{1}

\usepackage[margin=1in]{geometry}

\RequirePackage{amsthm,amsmath}
\usepackage{mathtools}
\usepackage{xr}
\usepackage{amssymb}
\usepackage{color}
\usepackage{algorithm}
\usepackage{algorithmic}
\usepackage{multicol}
\usepackage{booktabs}
\usepackage{cite}
\usepackage{setspace}

\usepackage[utf8]{inputenc}

\usepackage{textcomp,marvosym}

\usepackage{microtype}
\DisableLigatures[f]{encoding = *, family = * }

\DeclareMathOperator*{\argmin}{arg\,min}
\DeclareMathOperator*{\other}{other}
\DeclareMathOperator*{\upper}{upper}
\newtheorem{lemma}{Lemma}
\newtheorem{theorem}{Theorem}
\newtheorem{condition}{Condition}
\newtheorem{corollary}{Corollary}
\DeclareMathOperator*{\rem}{rem}

\begin{document}
\bibliographystyle{agsm}



\if1\blind
{
  \title{\bf Sparse-Input Neural Networks for High-dimensional Nonparametric Regression and Classification}
  \author{Jean Feng and Noah Simon
  	\thanks{
  		Jean Feng was supported by NIH grants DP5OD019820 and T32CA206089. Noah Simon was supported by NIH grant DP5OD019820.}
	\\
    Department of Biostatistics, University of Washington}
  \maketitle
} \fi

\if0\blind
{
  \bigskip
  \bigskip
  \bigskip
  \begin{center}
    {\LARGE\bf Sparse-Input Neural Networks for High-dimensional Nonparametric Regression and Classification}
\end{center}
  \medskip
} \fi

\bigskip
\begin{abstract}
Neural networks are usually not the tool of choice for nonparametric high-dimensional problems where the number of input features is much larger than the number of observations. Though neural networks can approximate complex multivariate functions, they generally require a large number of training observations to obtain reasonable fits, unless one can learn the appropriate network structure. In this manuscript, we show that neural networks can be applied successfully to high-dimensional settings if the true function falls in a low dimensional subspace, and proper regularization is used. We propose fitting a neural network with a sparse group lasso penalty on the first-layer input weights. This results in a neural net that only uses a small subset of the original features. In addition, we characterize the statistical convergence of the penalized empirical risk minimizer to the optimal neural network: we show that the excess risk of this penalized estimator only grows with the logarithm of the number of input features; and we show that the weights of irrelevant features converge to zero. Via simulation studies and data analyses, we show that these sparse-input neural networks outperform existing nonparametric high-dimensional estimation methods when the data has complex higher-order interactions.
\end{abstract}

\noindent%
\vfill

\newpage

	\section{Introduction}
	It is often of interest to predict a response $y$ from a set of inputs $x$.
	In many applications, the relationship between $x$ and $y$ can be quite complex and its functional form may be difficult to know apriori.
	In low and moderate dimensional settings, there are many methods that have been effective for estimating such complex relationships.
	For classification problems, popular methods include kernel extensions of Support Vector Machines, $k$-nearest neighbors, and classification trees \citep{cristianini2000introduction, kernel2008review};
	for regression problems, popular methods include spline regression, kernel regression and regression trees \citet{nadaraya1964estimating, watson1964smooth, breiman1984classification}.
	Neural networks have also proven to be particularly effective in problems from complex domains where other methods have had limited success (e.g. speech recognition, computer vision, and natural language processing, among others; \citet{graves2013speech, krizhevsky2012imagenet, szegedy2015going, socher2013recursive, mikolov2013distributed}).
	
	With the latest technological developments in biology and other fields, it is now very common to encounter high-dimensional data, where the number of features $p$ may far exceed the number of observations $n$.
	For example, in genome-wide, and epigenome-wide studies, it is common to collect thousands or millions of features on each of hundreds or thousands of subjects.
	For general problems in this setting, fully nonparametric methods like random forests or neural networks are rarely used since the number of training observations required for good performance is prohibitive.
	Instead, it is more typical to apply the Lasso \citep{tibshirani1996regression} or its additive non-parametric extensions such as Sparse Additive Models (SpAM) \citep{ravikumar2007spam} and high-dimensional additive models \citep{meier2009high}.
	These methods typically model the data using the sum of a small number of univariate (or very low-dimensional) functions.
	Unfortunately in actual scientific problems, the response may depend on complex interactions between multiple covariates, and failing to model these interactions can result in highly biased estimates.
	
	As existing estimation methods are ineffective for high-dimensional problems with complex interactions, we propose to address this gap using penalized neural networks.
	In particular, we leverage a unique quality of neural nets that sets them apart from more traditional nonparametric methods: With relatively few parameters, a neural net is able to approximate models with multivariate interaction terms \citep{barron1993universal}.
	This is in contrast to the exponential number of terms necessary required by polynomial or spline regression to model general multivariate functions.
	
	In this paper, we propose controlling the size of the neural network space using the sparse group lasso penalty, which is a mixed $\ell_1 / \ell_2$ penalty \citep{simon2013sparse}.
	Our method, sparse-input neural networks (SPINN), groups weights connected to the same input node and consequently selects a small subset of informative features for modeling a potentially complex response.
	We also provide a generalized gradient descent algorithm for training the network.
	
	To understand the theoretical properties of these sparse-input neural networks, we prove oracle inequalities that give probabilistic performance guarantees, assuming we have reached a global minimizer of our penalized criterion.
	We show that, if the response is best approximated by a sparse neural network that uses only $s$ of the features, then the difference between the prediction error of a sparse-input neural network and the prediction error of the best neural network shrinks at a rate of $O_p(n^{-1}s^{5/2} \log p)$ (here we treat the number of hidden nodes and layers as fixed).
	Hence the prediction error of sparse-input neural networks grows slowly with the number of features, making them suitable for high-dimensional problems.
	In addition, we show that the weights connected to the irrelevant input features also converge to zero.
	To our knowledge, there have been no theoretical results on the shrinkage rates of irrelevant model parameters for neural networks up to now.
	
	The paper is organized as follows.
	Section~\ref{sec:spinn} describes our extension of neural networks for high-dimensional problems and an algorithm to train the network.
	A discussion of work related to our method is provided in Section~\ref{sec:related work}.
	Section~\ref{sec:theory} establishes theoretical guarantees of our model.
	We present simulation studies in Section~\ref{sec:simulations} to better understand how SPINN behaves empirically.
	Finally, we analyze high- and moderate-dimensional data in Section~\ref{sec:data} and find that SPINN can significantly outperform more traditional nonparametric high-dimensional methods when the true function is composed of complex higher-order interaction terms.
	
	\section{Sparse-input neural networks}\label{sec:spinn}
	In this paper, we consider neural networks with a single output node and $L$ hidden layers, where hidden layer $a = 1,...,L$ has $m_a$ hidden nodes.
	For convenience, define $m_{L + 1} = 1$ and $m_0 = p$.
	The neural network parameters are denoted by
	$\boldsymbol{\eta} = \{(\boldsymbol{\theta}_a, \boldsymbol{t}_a)\}_{a=1}^{L+1}$ where hidden layer $a$ has weights $\boldsymbol{\theta}_a \in \mathbb{R}^{m_a \times m_{a-1}}$ and intercepts $\boldsymbol{t}_a \in \mathbb{R}^{m_a \times 1}$.
	The total number of parameters in the neural network is $q = \sum_{a=1}^{L+1} m_a (m_{a-1} + 1)$.
	In this paper, we suppose that all hidden nodes have the same activation function $\psi: \mathbb{R} \mapsto \mathbb{R}$, which we assume to be bounded and differentiable.
	For simplicity, $\psi(\boldsymbol{v})$ for a matrix $\boldsymbol{v}$ will denote applying $\psi$ to each element in $\boldsymbol{v}$.
	
	The output of the neural network is generated by propagating values between the neural network layers in the following recursive manner: the hidden node value at layer $a = 1,...,L$ is defined as
	\begin{align}
	\boldsymbol{z}_a(\boldsymbol{x}) &= \psi \left (
	\boldsymbol{\theta}_a \boldsymbol{z}_{a - 1}(\boldsymbol{x}) + \boldsymbol{t}_{a}
	\right)
	\end{align}
	where $\boldsymbol{z}_{0}(\boldsymbol{x}) = \boldsymbol{x} \in \mathbb{R}^{m_0 \times 1}$.
	For regression problems, the final neural network output is
	\begin{align}
	\label{eq:single-layer-nn}
	f_{\boldsymbol{\eta}}(\boldsymbol{x}) &= \boldsymbol{\theta}_{L + 1} \boldsymbol{z}_{L}(\boldsymbol{x}) + t_{L+1}.
	\end{align}
	A neural network for probabilistic binary classification is similar to \eqref{eq:single-layer-nn} except that we also apply the sigmoid function $\sigma(z) = 1/(1 + \exp(-z))$ at the end to ensure output values are between zero and one:
	\begin{align}
	\label{eq:single-layer-nn-class}
	f_{\boldsymbol{\eta}}(\boldsymbol{x}) &= \sigma \left(
	\boldsymbol{\theta}_{L + 1} \boldsymbol{z}_{L}(\boldsymbol{x}) + t_{L+1}
	\right ).
	\end{align}
	From henceforth in this manuscript, the upper layer network parameters refer to all the parameters above the first hidden layer, i.e. $\boldsymbol{\eta}_{\upper} = \{(\boldsymbol{\theta}_{a}, \boldsymbol{t}_a)\}_{a=2}^{L+1}$.
	
	By the Universal Approximation Theorem \citep{Leshno1993-bj, barron1993universal}, a neural network with a single hidden layer can approximate any function to any desired accuracy as long as the activation function is not a polynomial and there are a sufficient number of hidden nodes.
	Therefore neural networks can be thought of as a nonparametric estimator if we allow the number of hidden nodes/layers to grow as the number of observations increases.
	We may also think of neural networks nonparametric estimators from the viewpoint of sieve estimation: as we grow the number of hidden nodes in a network with a single hidden layer, the function class is monotonically increasing.
	
	Let $\ell(y, z): \mathbb{R} \times \mathbb{R} \mapsto \mathbb{R}^+$ denote the loss function, where $y$ represents the observed response and $z$ represents the prediction from the neural network.
	We suppose that $\ell$ is differentiable with respect to $z$.
	For regression problems, the loss function is typically the squared error loss $\ell(y,z) = (y - z)^2$.
	For classification problems, the loss function is typically the logistic loss $\ell(y,z) = -y\log z - (1 - y) \log (1 - z)$.
	
	Given observations $(\boldsymbol{x}_i, y_i)$ for $i = 1,...,n$, we propose fitting a sparse-input neural network where we penalize the first-layer weights using a sparse group lasso penalty and the upper-layer weights using a ridge penalty.
	Let $\boldsymbol{\theta}_{1, \cdot, j}$ denote the first-layer weights that are multiplied with the $j$th covariate (the $j$th column in the matrix $\boldsymbol{\theta}_1$).
	We fit SPINN by solving
	\begin{align}
	\begin{split}
	\label{eq:sgl_nn}
	\argmin_{\eta}
	& \frac{1}{n}
	\sum_{i=1}^n
	\ell \left (y_i, f_{\boldsymbol{\eta}}(\boldsymbol{x}_i) \right)
	+ \lambda_0 \sum_{a=2}^{L+1} \| \boldsymbol{\theta}_{a} \|^2_2
	+ \lambda \sum_{j=1}^p  \Omega_\alpha(\boldsymbol{\theta}_{1, \cdot, j})
	\end{split}
	\end{align}
	where $\alpha \in [0,1]$ and $\Omega_\alpha (\boldsymbol{\theta}) = (1-\alpha) \|\boldsymbol{\theta} \|_1 + \alpha \| \boldsymbol{\theta} \|_2$ is the sparse group lasso.
	Thus $\alpha$ can be thought of as a balancing parameter between the Lasso and the group lasso \citep{yuan2006model}: when $\alpha = 0$, the sparse group lasso reduces to the former and when $\alpha = 1$, it reduces to the latter.
	\begin{figure}
		\centering
		\includegraphics[width=0.5\textwidth]{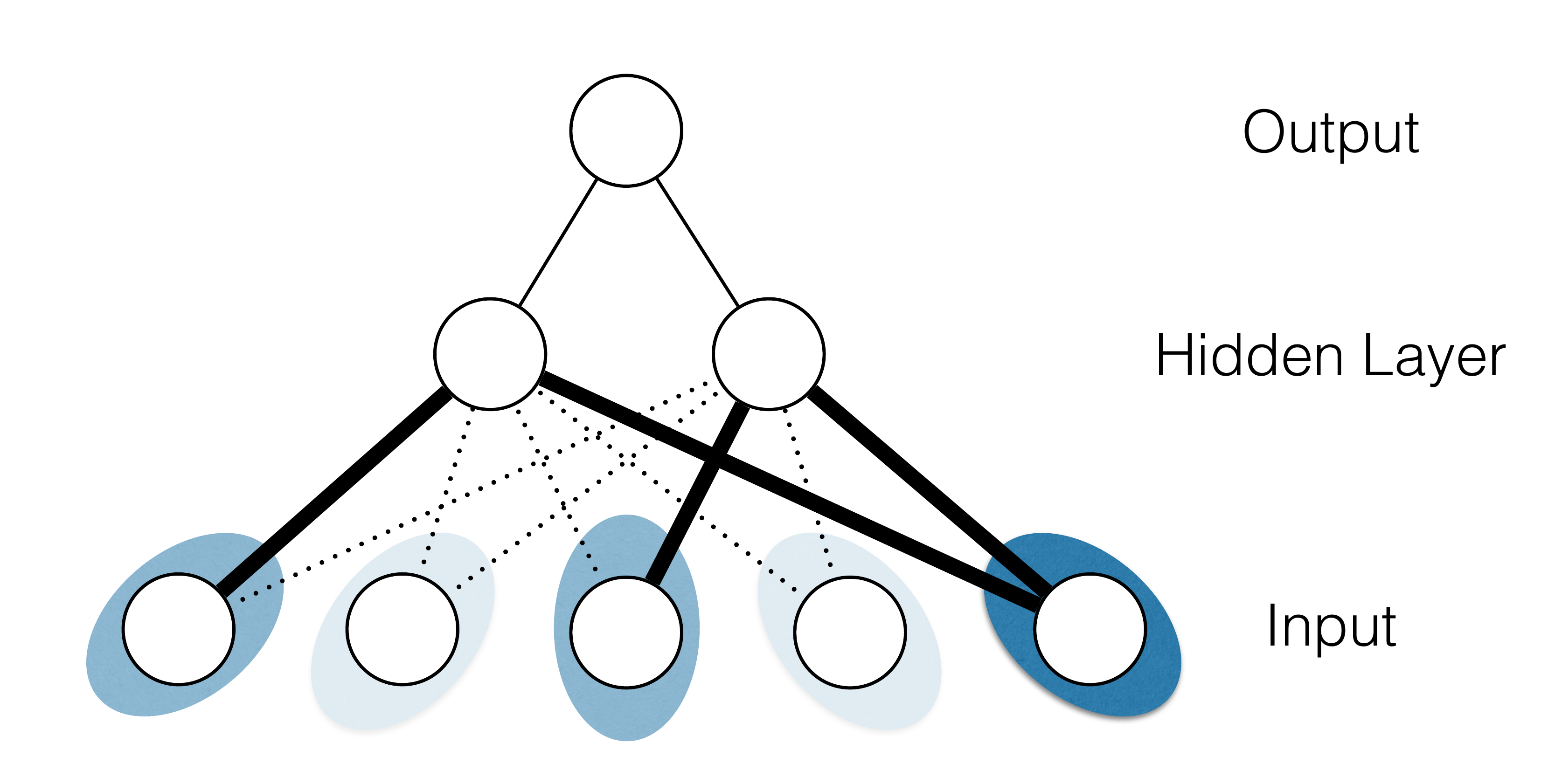}
		\caption{An example of a sparse-input neural network. The heavy and dotted lines indicate nonzero and zero weights, respectively. Each shaded oval corresponds to a group of first-layer weights. The weights from the dark blue oval are both nonzero. Each medium blue oval has a single nonzero weight, so they exhibit element-wise sparsity. All the weights in the light blue ovals are zero, so they exhibit group-level sparsity.}
		\label{fig:sgl_example}
		\vspace{-0.15in}
	\end{figure}
	
	We have three types of penalties in this criterion.
	The ridge penalty $\|\cdot\|^2_2$ serves to control the magnitude of the weights that are not in the first layer.
	The sparse group lasso $\Omega_\alpha $ is a mixture of the group lasso and lasso penalties \citep{simon2013sparse}.
	The group lasso penalty on $\boldsymbol{\theta}_{1,\cdot, j}$ encourages sparsity at the input level by encouraging the entire vector $\boldsymbol{\theta}_{1,\cdot, j}$ to be zero.
	The lasso penalty on $\boldsymbol{\theta}_1$ encourages sparsity across all the weights, so the hidden nodes are encouraged to connect to only a few input nodes.
	The parameter $\alpha$ allows us to control the level of sparsity at the level of the inputs vs. the individual weights.
	Pictorially, \eqref{eq:sgl_nn} encourages fitting neural networks like that in Figure~\ref{fig:sgl_example}.
	
	One could also consider adding a sparsity penalty to upper layer parameters.
	This is useful if the upper hidden layers have too many nodes and many need to be pruned away.
	However in the high-dimensional setting, performance is usually better for networks that have a small number of hidden layers and the upper hidden layers have a small number of nodes.

	\subsection{Learning}
	\label{sec:train}
	Sparse-input neural networks can be trained using generalized gradient descent \citep{daubechies2004iterative, beck2009fast, nesterov2013introductory}.
	Though generalized gradient descent was originally developed for convex problems, we can apply the work by \citet{gong2013general} to find a critical point in non-convex objective functions.
	Here we specialize their proposal, called GIST, to solve \eqref{eq:sgl_nn}.
	This algorithm is similar to that in \citet{alvarez2016learning}.
	
	Let $\mathcal{L}_{x,y}^{\text{smooth}}(\boldsymbol \eta)$ be the smooth component of the loss function in \eqref{eq:sgl_nn} defined as follows
	\begin{align}
	\mathcal{L}_{x,y}^{\text{smooth}}(\boldsymbol{\eta}) =
	\frac{1}{n}
	\sum_{i=1}^n
	\ell \left (
	y_i, f_{\boldsymbol{\eta}}(\boldsymbol{x}_i)
	\right )
	+ \lambda_0 \sum_{a=2}^{L+1} \| \boldsymbol{\theta}_a \|^2_2.
	\label{eq:smooth-loss}
	\end{align}
	Let $S(\cdot, \cdot): \mathbb{R}^p \times \mathbb{R} \mapsto \mathbb{R}^p$ be the coordinate-wise soft-thresholding operator
	\begin{align}
	\left (S\left(
	\boldsymbol{z}, c
	\right)\right )_j
	= \text{sign}(z_j) \left (
	|z_j| - c
	\right )_+.
	\end{align}
	
	The algorithm for training a sparse-input neural network is given in Algorithm~\ref{algo:gen_grad_descent}. The proximal gradient step is composed of three sub-steps. The first sub-step \eqref{eq:backprop} performs a gradient update step only for the smooth component of the loss; the gradient can be computed using the standard back-propagation algorithm. The second and third sub-steps, \eqref{eq:prox1} and \eqref{eq:prox2}, are the proximal operations for the Sparse Group Lasso \citep{simon2013sparse}: a soft-thresholding operation followed by a soft-scaling operation on each $\boldsymbol{\theta}_{1,\cdot,j}$.
	Therefore training a sparse-input neural networks should take a similar amount of time as training a traditional neural network where the objective function only contains a ridge penalty.
	
	At each iteration of the algorithm, the step size $\gamma_k$ is initialized with some value in $[\gamma_{\min}, \gamma_{\max}]$ and then tuned according to a monotone line search criterion.
	In our implementation, we simply used a fixed value, e.g. $\gamma_{min} = \gamma_{max}$, though one can also adaptively choose the initial step size \citep{barzilai1988two, gong2013general}.
	Let $\boldsymbol{\eta}^{(k - 1, 2)}$ be the model parameters at the $k - 1$-th iteration and  $\boldsymbol{\eta}^{(k , 2)}$ be the proposed model parameters for step size $\gamma_k$ at iteration $k$.
	The monotone line search criterion accepts step size $\gamma_k$ if the following condition is satisfied:
	\begin{align}
	L(\boldsymbol{\eta}^{(k,2)}) \le L(\boldsymbol{\eta}^{(k - 1,2)}) - t \gamma_k \|\boldsymbol{\eta}^{(k,2)} - \boldsymbol{\eta}^{(k - 1,2)} \|^2
	\end{align}
	where $L(\cdot)$ is the objective function of \eqref{eq:sgl_nn} and $ t \in (0,1)$.
	This line search criterion guarantees that Algorithm~\ref{algo:gen_grad_descent} converges to a critical point (where the subdifferential contains zero) \citep{gong2013general}. 
	
	Finally, another option for training sparse-input neural networks is to apply the accelerated generalized gradient descent framework for non-convex objective functions developed in \citet{ghadimi2016accelerated}. These are guaranteed to converge to a critical point, and have accelerated rate guarantees.
	
	\begin{algorithm}
		\caption{Training sparse-input neural networks}
		\label{algo:gen_grad_descent}
		\begin{algorithmic}
			\STATE{
				Initialize neural network parameters $\boldsymbol{\eta}^{(0,2)}$. Choose $s \in (0,1)$ and $\gamma_{min},\gamma_{max}$ such that $\gamma_{max} \ge \gamma_{min} > 0$.
			}
			\FOR{iteration $k=1,2,...$}
			\STATE 	$\gamma_k \in \left[\gamma_{min}, \gamma_{max} \right]$
			\REPEAT
			\vspace{-0.2in}
			\STATE \begin{align}
			&\boldsymbol{\eta}^{(k,0)}
			= \boldsymbol{\eta}^{(k-1,2)} - \gamma_k \nabla_\eta \mathcal{L}_{x,y}^{\text{smooth}}(\boldsymbol{\eta}^{(k - 1,2)})
			\label{eq:backprop}
			\\
			& \boldsymbol{\theta}^{(k,1)}_1 = S\left ({\boldsymbol\theta}^{(k,0)}_1, \gamma_k \lambda (1 -\alpha) \right)
			\label{eq:prox1}
			\end{align}
			\FOR{$j = 1,..., p$}
			\vspace{-0.15in}
			\STATE
			\begin{align}
			\boldsymbol{\theta}^{(k,2)}_{1,\cdot,j} & = \left(
			1 - \frac{\gamma_k \lambda \alpha}{\left\|\boldsymbol{\theta}^{(k,1)}_{1,\cdot,j}\right\|_2}
			\right)_+
			{\boldsymbol\theta}^{(k,1)}_{1,\cdot,j}
			\label{eq:prox2}
			\end{align}
			\vspace{-0.15in}
			\ENDFOR
			\FOR{$a = 2,..., L+1$}
			\STATE
			\begin{align*}
			\left(
			\boldsymbol{\theta}^{(k,2)}_a,
			\boldsymbol{t}^{(k,2)}_a
			\right)
			= \left(
			\boldsymbol{\theta}^{(k,0)}_a,
			\boldsymbol{t}^{(k,0)}_a
			\right)
			\end{align*}
			\ENDFOR
			\STATE $\gamma_k \coloneqq s \gamma_k$
			\UNTIL{
				line search criterion is satisfied
			}
			\ENDFOR
		\end{algorithmic}
	\end{algorithm}
	
	\subsection{Tuning Hyper-parameters}
	There are two types of hyper-parameters for fitting a sparse-input neural network: the penalty parameters for the ridge and sparse group lasso penalties and those specifying the number of layers and number of nodes per layer in the network.
	The hyper-parameters should be tuned to ensure low generalization error of the model and are typically chosen via cross-validation.
	
	When fitting sparse-input neural networks for high- or moderate-dimensional problems, we found that the hyper-parameter values can be tuned using either grid search or gradient-free optimization methods such as Nelder-Mead \citep{nelder1965simplex} or Bayesian optimization \citep{snoek2012practical}.
	The former method is typically used when there are three or fewer hyper-parameters and the latter class is typically used when there are many hyper-parameters.
	In our implementation, we typically consider a number of possible neural network structures and different penalty parameter values; grouping the hyper-parameters in this manner result in four hyper-parameters to tune in SPINN.
	Thus the number of hyper-parameters is only slightly more than the typical use case in grid search.
	
	The computation time for tuning hyper-parameters can be further reduced by pretuning the range of the possible hyper-parameter values.
	In particular, the optimal network structure for such high-dimensional problems tends to be small since large networks easily overfit to the data.
	For the examples in this paper, we only consider network structures with no more than three hidden layers and no more than fifty hidden nodes per layer.
	In addition, since the upper layers of the network are parameterized by a small number of parameters, we found that the generalization error is relatively insensitive to the ridge penalty parameter.
	For all of our empirical analyses, we pretune the ridge penalty parameter and fix it to a small value (usually $\le$ 0.001).
	
	The rest of the hyper-parameters play a more influential role on the generalization error.
	Understanding their affects on the performance not only increases our understanding of SPINN but also provides guidance for tuning their values.
	We analyze the role of the network structure and the sparse group lasso penalty parameter $\lambda$ from a theoretical standpoint in Section~\ref{sec:theory} and present an empirical analysis of the role of the balancing parameter $\alpha$ in Section~\ref{sec:simulations}.

	\section{Related Work}\label{sec:related work}
	
	A number of other authors have applied lasso and group lasso penalties to neural networks.
	That work has largely been focused on learning network structure however, rather than feature selection.
	\citet{sun1999lasso} was one of the earliest papers that fit a neural network with a lasso penalty over all the weights.
	\citet{scardapane2016group} and \citet{alvarez2016learning} proposed using the sparse group lasso over all the weights in a deep neural network in order to learn a more compact network --- this is in the low dimensional setting with $n \gg p$.
	The recent work by \citet{JMLR:v18:14-546} is most closely aligned to our work: he considers convex neural networks which have a single hidden layer with an unbounded number of hidden nodes.
	He shows that theoretically, a lasso penalty on the input weights of such networks should perform well in high-dimensional settings.
	
	The recent work in \citet{zhang2016understanding} examined the utility of regularization in deep learning. The authors found that using a ridge penalty on the internal network architecture was not necessary for a neural network to have good performance; instead, using an appropriate network architecture led to larger decreases in generalization error. Our results support this claim since the sparse group lasso also performs structure-learning.
	
	Previously, statistical convergence rates for neural networks have been established for problems when the $\ell_1$-norm of the weights are constrained.
	These results show that the estimation error of the neural networks indeed grew with the logarithm of the number of input nodes $p$ \citep{bartlett1998sample, Anthony2009-hr}.
	Our convergence rate results also have a $\log p$ term.
	However, we improve upon previously established convergence rates by showing that the excess estimation error shrinks at a rate of $n^{-1}$ rather than $n^{-1/2}$ as was given in previous bounds.
	This faster convergence rate allows us to additionally bound the rate at which the norm of the irrelevant network weighs shrink to zero.
	
	Our proofs for statistical convergence rates are inspired by \citet{stadler2010}, which considers $\ell_1$-penalization for mixture regression models.
	The techniques used in their paper are relevant as neural networks can be thought of as a mixture of regressions.
	Significant additional work was required however as the identifiability condition assumed in \citet{stadler2010} does not hold for neural networks.

	\section{Theoretical Guarantees}
	\label{sec:theory}
	
	In this section, we provide probabilistic, finite-sample upper bounds on the prediction error of sparse-input neural networks.
	Instead of using a ridge penalty, we constrain the norm of the upper-layer weights.
	Our sparse-input neural network problem now becomes
	\begin{align}
	\begin{split}
	\label{eq:lasso_nn}
	\hat{\boldsymbol{\eta}} &  \in
	\argmin_{\eta \in \Theta}
	\frac{1}{n}
	\sum_{i=1}^n
	\ell \left (
	y_i, f_{\eta}(\boldsymbol{x}_i)
	\right )
	+ \lambda
	\sum_{j=1}^p
	\Omega_\alpha(\boldsymbol{\theta}_{1,\cdot,j})
	\\
	\text{where }
	\Theta & = \left\{
	\boldsymbol{\eta} \in \mathbb{R}^{q}:
	\|\other({\boldsymbol{\eta}})\|_2 \le K
	\right\}
	\end{split}
	\end{align}
	for a constant $K > 0$ and $\other({\boldsymbol{\eta}}) = (\boldsymbol{t}_1, \{(\boldsymbol{\theta}_a, \boldsymbol{t}_a) \}_{a=2}^{L+1})$.
	All of our proofs are in the Supplementary Materials.
	
	Notice that \eqref{eq:lasso_nn} assumes that the estimator is a global minimizer of a non-convex objective function. Since non-convex problems are typically computationally intractable to solve,  there is admittedly a disconnect between the computational algorithm we have proposed and the theoretical results we establish in this section. Though it is desirable to establish theoretical properties for estimators arising from local optima, it is difficult to characterize their behavior and up to now, much of the theory for the Lasso depends on the estimator being a global minimizer. We do not address this issue and leave this problem for future research.
	
	\subsection{Problem Setup and Notation}
	Suppose covariates $X$ have support $ \mathcal{X} \subseteq [-X_{\max}, X_{\max}]^p$ and $\mathcal{X}$ contains some open set.
	In addition, suppose $Y = f^*(X) +\epsilon$ where $\epsilon$ is a random variable with mean zero.
	Let $\mathbb{P}$ denote the expectation with respect to the joint distribution of the covariates $X$ and response $Y$.
	Given $n$ observations, we denote the empirical distribution as $\mathbb{P}_n$.
	
	The neural network can be defined either as in \eqref{eq:single-layer-nn} or \eqref{eq:single-layer-nn-class}, where we suppose the activation function $\psi$ is the hyperbolic tangent function.
	Note that neural networks with the sigmoid function are included in this framework since $\tanh$ and the sigmoid function are related by a linear transformation.
	
	Next we define a neural network equivalence class.
	Given parameter $\boldsymbol{\eta}$, the set of equivalent parameterizations is
	\begin{align}
	EQ(\boldsymbol{\eta}) =
	\left\{ \boldsymbol{\eta}^{'} \in \Theta : f_{\boldsymbol{\eta}^{'}}(\boldsymbol{x}) = f_{\boldsymbol{\eta}}(\boldsymbol{x}) \forall \boldsymbol{x} \in \mathcal{X} \right \}.
	\end{align}
	If $\boldsymbol{\eta}$ is ``generic'' as defined in \citet{Fefferman1994-kz}, $EQ(\boldsymbol{\eta})$ only contains parameterizations that are sign-flips or permutations of $\boldsymbol{\eta}$ and thus $EQ(\boldsymbol{\eta})$  has cardinality of at most $\prod_{a=1}^{L} 2^{m_a} (m_a!)$.
	Roughly, a generic NN parameter is one where all hidden nodes have nonzero input weights, $\boldsymbol{\theta}_{a,j, \cdot} \ne \boldsymbol{0}$, and no two hidden nodes are sign-flips of each other, i.e.  $|\boldsymbol{\theta}_{a,j', \cdot}| \ne |\boldsymbol{\theta}_{a,j, \cdot}|$ for $j \ne j'$.
	
	Let the set of optimal neural networks that minimize the expected loss be denoted
	\begin{align}
	EQ^* = \argmin_{\boldsymbol{\eta} \in \Theta} \mathbb{P} \ell(y, f_{\boldsymbol{\eta}}(\boldsymbol{x})).
	\label{eq:exp_loss}
	\end{align}
	We suppose that $EQ^*$ is the union of $Q \ge 1$ equivalence classes, where all of the optimal neural networks are generic.
	We will suppose that $\Theta$ is chosen large enough such that the gradient of the expected loss at every element in $EQ^*$ is zero.
	
	Next suppose that the expected loss is minimized using neural networks that employ only a few input features $S$, which we call the ``relevant'' features.
	That is, for $\boldsymbol{\eta}^* \in EQ^*$, we suppose that all weights tied to the irrelevant features $S^c$ are zero.
	Note that all neural nets from the same equivalence class have the same set of features with nonzero weights because members of the equivalence class are permutations and/or sign-flips of one another.
	In this work we are particularly interested in the case where the optimal neural networks are sparse, e.g. those where the cardinality of $S$, denoted $|S|$, is small.
	For any $\boldsymbol{\eta} \in \Theta$, let $\boldsymbol{\theta}_{1,\cdot,S}$ denote the weights tied to the input nodes $S$ and $\boldsymbol{\theta}_{1,\cdot,S^c}$ denote the weights tied to the input nodes $S^c$.
	Let $\boldsymbol{\eta}_S$ denote the NN parameters that are exactly the same as $\boldsymbol{\eta}$ except where $\boldsymbol{\theta}_{1,\cdot,S^c}$ is set to zero.
	
	For any $\boldsymbol{\eta}$, let the closest element in $EQ^*$  be defined as
	$$
	\boldsymbol{\eta}^{*(\boldsymbol{\eta})} =
	\{(
	\boldsymbol{\theta}_a^{*(\boldsymbol{\eta})},
	\boldsymbol{t}_a^{*(\boldsymbol{\eta})}
	)\}_{a=1}^{L+1}
	\in \argmin_{\boldsymbol{\eta}^* \in EQ^*}
	\left \| \boldsymbol{\eta} - \boldsymbol{\eta}^* \right \|_{2},
	$$
	where we randomly pick one optimal network if there are multiple that minimize the distance to $\boldsymbol{\eta}$.
	Define the excess loss of a neural network with parameters $\boldsymbol{\eta}$ as
	\begin{align}
	\mathcal{E}(\boldsymbol{\eta}) &=
	\mathbb{P} \ell(y,f_{\boldsymbol{\eta}}(\boldsymbol{x}))
	- \min_{\boldsymbol{\eta}^{'} \in \Theta} \mathbb{P} \ell(y,f_{\boldsymbol{\eta}^{'}}(\boldsymbol{x}))
	\\
	& = \mathbb{P} \left [
	\ell(y,f_{\boldsymbol{\eta}}(\boldsymbol{x})) -
	\ell(y,f_{\boldsymbol{\eta}^{*(\boldsymbol{\eta})}}(\boldsymbol{x}))
	\right ].
	\label{eq:excess}
	\end{align}
	
	\subsection{Results}
	To understand the behavior of our estimated neural network from \eqref{eq:lasso_nn}, we upper bound the excess loss as well as the norm of the weights connected to the irrelevant inputs.
	Our proof technique is inspired by  \citet{stadler2010}; however significant adaptations were needed to deal with the complex loss surface and equivalence classes of neural networks.
	
	In order for our results to hold, we make the assumption that the expected loss is locally strongly convex at all $\boldsymbol{\eta}^{*} \in EQ^{*}$.
	Since this only specifies the local behavior at $EQ^{*}$, this assumption is relatively weak.
	We use the notation $A \succeq B$ to indicate that $A - B$ is a positive semi-definite matrix.
	More formally this assumption states:
	\begin{condition}
		\label{cond:hess}
		Let the neural network parameter $\boldsymbol{\eta}$ be ordered such that
		$$
		\boldsymbol{\eta} = \left (
		\boldsymbol{\theta}_{1,\cdot,S^c}, \boldsymbol{\theta}_{1,\cdot,S}, \boldsymbol{t}_1,
		\boldsymbol{\theta}_2, \boldsymbol{t}_2,
		\cdots,
		\boldsymbol{\theta}_{L+1}, \boldsymbol{t}_{L+1}
		\right ).
		$$
		There is a constant $h_{min}>0$ that may depend on $m_1,...,m_{L+1}, |S|, f^*$ and the distribution $\mathbb{P}$, but does not depend on $p$,
		such that for all $\boldsymbol{\eta}^* \in EQ^*$,
		\begin{align}
		\left[\nabla_{\eta}^{2}\mathbb{P}\ell(y, f_{\boldsymbol{\eta}}(\boldsymbol{x}))\right]_{\eta=\eta^*}
		& \succeq
		h_{min}
		\left[\begin{array}{cc}
		\boldsymbol{0} & \boldsymbol{0}\\
		\boldsymbol{0} & \boldsymbol{I}
		\end{array}\right]
		\end{align}
		where the top left zero matrix is $m_1|S^c| \times m_1|S^c|$ and the bottom right matrix is the identity matrix of dimension $(q - m_1|S^c|) \times (q - m_1|S^c|)$.
	\end{condition}
	
	In addition, we need the following identifiability condition.
	\begin{condition}
		\label{cond:identifiable}
		For all $\epsilon>0$,
		there is an $\chi_{\epsilon}>0$ that may depend on $m_1,...,m_{L+1}, |S|, f^*$ and the distribution $\mathbb{P}$, but does not depend on $p$, such that
		\begin{align*}
		\begin{split}
		\chi_{\epsilon} & \le \inf_{\boldsymbol{\eta} \in \Theta} \left\{ \mathcal{E}\left(\boldsymbol{\eta}\right):
		\|\boldsymbol{\eta}_S - \boldsymbol{\eta}^{*(\boldsymbol \eta)}\|_2
		\ge\epsilon  \text{ and }
		\right .
		\\
		& \left .
		\|\boldsymbol \theta_{1,\cdot,S^c}\|_1 \le
		3 \sum_{j \in S}
		\Omega_\alpha (\boldsymbol \theta_{1,\cdot,j} - \boldsymbol \theta_{1,\cdot,j}^{*(\boldsymbol \eta)}) +
		\|\other(\boldsymbol{\eta}) - \other(\boldsymbol{\eta}^{*(\boldsymbol{\eta})})\|_2
		\right\}.
		\end{split}
		\end{align*}
	\end{condition}
	Condition~\ref{cond:identifiable} places a lower bound on the excess loss of neural networks outside the set of optimal neural networks $EQ^*$.
	However we only need this lower bound to apply to neural networks where the weight of the irrelevant nodes is dominated by the difference between the other parameters in $\boldsymbol{\eta}$ and $\boldsymbol{\eta}^{*(\boldsymbol \eta)}$.
	By restricting to this smaller set of neural networks, it is more realistic to claim that $\chi_{\epsilon}$ is independent of $p$.
	This condition is similar to compatibility conditions used in proving theoretical results for the Lasso \citep{buhlmann2011statistics}.
	
	Finally, we require a bound on the third derivative of the expected loss.
	Since $\mathcal{X}$ and the upper layer weights are bounded, It is easy to show that this condition is satisfied when the loss function $\ell$ is mean squared error or logistic loss.
	\begin{condition}
		\label{condn:bound_third}
		The third derivative of the expected loss function is bounded uniformly over $\Theta$ by some constant $G > 0$ that may depend on $m_1,...,m_{L+1}, |S|, f^*, K$ and the distribution $\mathbb{P}$, but does not depend on $p$:
		\begin{equation}
		\sup_{\boldsymbol{\eta} \in \Theta}\max_{j_{1}j_{2}j_{3}}
		\left|
		\frac{\partial^{3}}{\partial\eta_{j_{1}}\partial\eta_{j_{2}}\partial\eta_{j_{3}}}
		\mathbb{P}\ell(y, f_{\boldsymbol{\eta}}(\boldsymbol{x}))
		\right|\le G.
		\label{eq:third_deriv_bound}
		\end{equation}
	\end{condition}

	With the three conditions above, we have the following theorem.
	We use the notation $a \vee b = \max(a,b)$.
	\begin{theorem}
		\label{thm:sparse_nn}
		For any $\tilde{\lambda}>0$ and $T\ge1$, let 
		\begin{align*}
		\begin{split}
		\mathcal{T}_{\tilde{\lambda},T}=&  \Big \{ \left \{ \left(\boldsymbol{x}_{i},y_{i}\right) \right \} _{i=1}^{n}:
		\\
		&
		\left .
		\sup_{\boldsymbol{\eta} \in \Theta}
		\frac{
			\left|\left(\mathbb{P}_{n}-\mathbb{P}\right)
			\left(
			\ell(y, f_{\eta^{*(\boldsymbol \eta)}}(\boldsymbol{x}))
			-\ell(y, f_{\boldsymbol \eta}(\boldsymbol{x}))
			\right)\right|
		}
		{
			\tilde{\lambda} \vee
			\left (
			\left\Vert \other(\boldsymbol{\eta}) - \other(\boldsymbol{\eta}^{*(\boldsymbol{\eta})}) \right\Vert _{2}
			+ \sum_{j=1}^p \Omega_\alpha(\boldsymbol{\theta}_{1,\cdot,j}-\boldsymbol{\theta}_{1,\cdot,j}^{*(\boldsymbol \eta)}) \right )
		}
		\le T\tilde{\lambda}  \right \} .
		\end{split}
		\end{align*}
		Suppose Conditions~\ref{cond:hess}, \ref{cond:identifiable}, and \ref{condn:bound_third} hold.
		Suppose that the optimal neural networks $EQ^*$ only have nonzero weights for input features in $S$.
		Let $\hat{\boldsymbol{\eta}}$ be a solution of \eqref{eq:lasso_nn}.
		Then over the set $\mathcal{T}_{\tilde \lambda, T}$, we have for any $\lambda\ge 2\tilde{\lambda}T$
		\begin{align}
		\label{eq:thrm_oracle_ineq}
		\mathcal{E}\left(\hat{\boldsymbol\eta}\right)+
		2\left(\lambda-T\tilde{\lambda}\right) \sum_{j \in S^c} \Omega_\alpha(\hat{\boldsymbol \theta}_{1,\cdot,j})
		\le
		\left(T\tilde{\lambda}+\lambda\right)^{2}
		((1-\alpha)\sqrt{m_1} + \alpha )^2|S| C_{0}^{2}
		\end{align}
		where
		\begin{align}
		C_{0}^{2} & = \frac{1}{\epsilon_{0}} \vee \frac{C^2_1 (K + \max_{\eta^* \in EQ^*} \sum_{j \in S} \Omega_\alpha(\boldsymbol \theta_{1,\cdot,j^*}))^{2}}{\chi_{\epsilon_{0}}},\\
		\epsilon_{0} & =\frac{h_{min} (1 -\alpha + \alpha/\sqrt{m_1})^3}{C_1 G \left(
			(1 - \alpha) \sqrt{m_1|S|}  + \alpha \sqrt{|S|}
			+ \sqrt{m^*}
			\right)^3}
		\end{align}
		for some constant $C_1 > 0$.
	\end{theorem}
	Theorem~\ref{thm:sparse_nn} simultaneously bounds the excess loss and the norm of the irrelevant nodes.
	Consider the case where the identifiability constant $\chi_{\epsilon_{0}}$ is sufficiently large such that $C_0^2 = \epsilon_0^{-1}$.
	Then the above theorem states that for $\alpha \in (0,1)$ (e.g. we are using the sparse group lasso and not the degenerate cases with only the lasso or only the group lasso), the excess loss will be on the order of $O_p(\tilde \lambda^2 G m_1^{5/2} |S|^{5/2} )$
	and the norm of $\hat {\boldsymbol{\theta}}_{S^c}$ will shrink at the rate of $O_p(\tilde \lambda  G m_1^{5/2} |S|^{5/2})$.
	The convergence rate of the excess risk is faster for functions that are best approximated by neural networks that are more sparse (e.g. $|S|$ is small).
	If $\tilde{\lambda}$ shrinks as the number of samples increases, these values will go to zero.
	The rate of convergence is fastest if we set $\tilde{\lambda}$ to a very small value; however we must choose $\tilde{\lambda}$ carefully so that $\mathcal{T}_{\tilde{\lambda},T}$ occurs with high probability.
	
	Our convergence rate of here depends on the number of hidden nodes in the first layer and implicitly depends on the upper layers contribute to the rate through the constant $G$ in Condition~\ref{condn:bound_third}.
	If we want the overall estimation error, we must also take into account how the network structure contributes via the approximation error.
	Increasing the number of nodes and layers in the network will decrease the approximation error; however the variance of model will increase and our upper-bound in \eqref{eq:thrm_oracle_ineq} will be large.
	
	The relevance of Theorem~\ref{thm:sparse_nn} depends on the probability of the set ${\mathcal{T}}_{\tilde{\lambda}, T}$.
	To analyze this sets probability, we introduce one more condition that bounds the gradient of the loss with respect to the nodes in the \textit{first hidden layer} and \textit{the network parameters in the upper layers}.
	Similar to Condition~\ref{condn:bound_third}, it is easy to show that this condition is satisfied when the loss function is mean squared error or logistic loss.
	\begin{condition}
		\label{condn:bound_first}
		Define the upper portion of the neural network $f_{\eta}$ as $f_{\eta_{\upper}}: \mathbb{R}^{m_1} \mapsto \mathbb{R}$.
		The gradient of the loss of $f_{\boldsymbol{\eta}_{\upper}}$ is uniformly bounded by:
		\begin{equation}
		\sup_{\boldsymbol{\eta} \in \Theta, \boldsymbol{x} \in \mathcal{X}}
		\left\|
		\left .
		\nabla_{\boldsymbol{z}, \boldsymbol{\eta}_{\upper}}
		\ell(y, f_{\boldsymbol{\eta}_{\upper}}(\boldsymbol{z}))
		\right |_{\boldsymbol{z} = \psi(\boldsymbol{\theta}_1^\top \boldsymbol{x} + \boldsymbol{t})}
		\right\|_\infty
		\le M_1(|y| + M_0)
		\label{eq:first_deriv_bound}
		\end{equation}
		where constants $M_0, M_1 > 0$ may depend on $m_1,...,m_{L+1}, |S|, f^*, K,$ and $X_{\max}$, but does not depend on $p$.
	\end{condition}
	
	We specialize the following theorem to classification and regression settings.
	The proof relies on techniques from empirical process theory.
	Let $m_{\upper}$ be the number of network parameters in the upper layers of the network (i.e. exclude the weights in the first layer).

	\begin{theorem}
		\label{thm:emp_proc}
		Consider the following two settings:
		\begin{enumerate}
			\item Regression setting: Let $$
			y = f^*(\boldsymbol{x}) + \epsilon
			$$
			where $\epsilon$ is a sub-gaussian random variable with mean zero. That is, $\epsilon$ satisfies for some constants $\tau$ and $\sigma_0$,
			\begin{align}
			\tau^2(\mathbb{E} e^{|\epsilon|^2/K_\epsilon^2} - 1) \le \sigma_0^2.
			\end{align}
			We suppose that $\epsilon$ is independent of $\boldsymbol{x}$.
			Suppose $\ell(\cdot,\cdot)$ is squared-error loss.
			\item Classification setting: Let $Y$ take on values $\{0,1\}$ where $p(Y = 1|\boldsymbol{x}) = f^*(\boldsymbol{x})$.
			Suppose $\ell(\cdot,\cdot)$ is logistic loss.
		\end{enumerate}
		Suppose that the set of optimal neural networks $EQ^*$ is composed of $Q$ equivalence classes.
		For each setting, there exists a constant $c_0 > 0$ that only depends on $\tau$ and $\sigma_0$ such that for
		\begin{align}
		\tilde{\lambda}=
		c_0 M_1
		(m_1 + K \sqrt{m_{\other}})
		\sqrt{\frac{\log n}{n}}
		\left (
		\sqrt{\log Q} + \frac{X_{\max}}{c_{1}}\log (n c_{2}) \sqrt {\log(c_{2} p)}
		\right ).
		\label{eq:tilde_lam}
		\end{align}
		we have for any $T\ge1$,
		\begin{align}
		Pr_{X,Y}\left(\mathcal{T}_{\tilde{\lambda},T}\right) \ge 1 - O\left(\frac{1}{n} \right )
		\label{eq:high_prob}
		\end{align}
		where $c_{1} = 1 - \alpha + \alpha/\sqrt{m_1}$ and $c_{2} = m_1 + K\sqrt{m_{\other}} + X_{\max}/c_{1}$.
	\end{theorem}
	
	We now combine Theorems~\ref{thm:sparse_nn} and \ref{thm:emp_proc} to obtain a final bound on the excess loss and norm of the irrelevant weights.
	Here we focus primarily on the role of $p$ rather than the role of the number of layers and hidden nodes since in high-dimensional settings, we find shallow networks typically have better performance.
	In the above results, the effect of the number of layers and hidden nodes in the upper layers are encapsulated in the constants $G$ and $M_1$, which are used to bound the first and third derivatives of the loss over the entire space of $\Theta$.
	(It is easy to upper-bound these constants with a rate that grows exponentially in the number of hidden layers, which is similar to previous results \citep{Anthony2009-hr}.)
	Thus to combine results, define $M_2 = G M_1 (\sqrt{m_{\other}} + m_1)$, which depends on the number of hidden nodes at the hidden layers, $|S|$, $f^*$, $K$ and $X_{\max}$ but not $p$.
	If the identifiability constant is not too small, Theorems~\ref{thm:sparse_nn} and \ref{thm:emp_proc} state that if $\alpha \in (0,1)$ and  $\tilde{\lambda}$ is chosen according to \eqref{eq:tilde_lam}, the excess loss converges at the rate
	$$
	O_p\left(
	n^{-1} M_2^2 m_1^{5/2} |S|^{5/2} \log p
	\right)
	$$
	and the $\ell_1$-norm of the irrelevant weights converges at the rate
	$$
	O_p\left(
	n^{-1/2} M_2 m_1^{5/2} |S|^{5/2}  \sqrt{\log p}
	\right)
	$$
	modulo log terms that do not depend on $p$.
	Thus we find that the total number of features $p$ only enters these rates in a log-term.
	
	To the best of our knowledge, our results are the first to bound the convergence rate of the norm of irrelevant weights.
	Though our bounds might not be tight, they help us understand why we observe better performance in sparse-input neural networks compared to standard neural networks in our simulation studies and data analyses.
	
	We hope to improve these convergence rate bounds in future research.
	In particular, we would like to shrink the exponent on $m_1$ and $|S|$ since these rates become very slow for moderate values of $m_1$ and $|S|$.
	However we find that in practice, the number of effective nodes in the first hidden layer is controlled partly by the shrinkage behavior of the lasso since the lasso encourages sparsity both in the inputs as well as the first hidden layer.
	Thus the convergence rates are likely faster when the true function $f^*$ is closely approximated by a network with a small number of hidden nodes in the first layer.
	Also, our bound depends on the gradient of the loss of the networks in the model class.
	It seems possible to tighten the bounds if we can show that the gradient of the fitted network is similar to that of the optimal neural network.
	Intuitively, if the optimal neural network is smooth, i.e. the magnitude of these derivatives are small, then the fitted network by optimizing the penalized log likelihood is also likely to be smooth.

	\section{Simulation study}
	\label{sec:simulations}
	
	In this section, we present simulations to understand various facets of these sparse-input neural networks.
	First, we compare empirical convergence rates to our theoretical convergence rates to understand if our rates hold and how tight they are.
	Second, we would like to understand the roles of the lasso vs. group lasso penalties and how they interact with each other.
	Finally, we compare SPINN to other methods in various settings.
	
	For all of the simulations, we generated the data according to the model
	$ y = f^*(\boldsymbol{x}) + \sigma \epsilon $ where $\epsilon \sim N(0,1)$ and $\sigma$ is scaled such that the signal to noise ratio is 2.
	We sampled each covariate $x_j \in \mathbb{R}$ independently from the standard uniform distribution.
	For the simulations, the true function $f^*$ will only depend on the first $|S|$ features, though the number of covariates $p$ may be much larger.
	Throughout the simulations, the mean squared error is defined as $E[(f^*(x) - f_{\hat{\boldsymbol{\eta}}}(x))^2]$ since the true model is known.
	
	\subsection{Confirming the convergence rates}
	Here we aim to understand if our convergence rates hold in practice.
	Recall that our theoretical results depend on a number of assumptions.
	We assumed that the fitting procedure finds the global minima, even though this is difficult to show in practice since our optimization problem is not convex.
	In addition, Conditions~\ref{cond:hess} and \ref{cond:identifiable} assumed that the behavior of the expected loss depended on constants $h_{\min}$ and $\alpha_\epsilon$ that did not depend on $p$.
	
	As our theoretical results only bound the excess loss, we would like an experimental setup that is not affected by the approximation error.
	Here we set the true function to be a neural network with one hidden layer and four hidden nodes:
	\begin{align*}
	f^*(\boldsymbol{x}) &= \tanh(x_1 + 2x_2 - 3x_3 + 2x_4)  + 2\tanh(x_1 - 2 x_5 + 2 x_6)\\
	& + \tanh(-x_2 - x_3 - x_6) + \tanh(x_5 - 0.5 x_3 + 0.5 x_6).
	\end{align*}
	Thus the excess loss is $E[(f^*(x) - f_{\hat{\boldsymbol{\eta}}}(x))^2]$.
	
	First we check the behavior of SPINN as we grow the number of observations.
	Here the fitted neural networks have the same structure as the optimal neural network and we keep the parameter $\alpha$ fixed.
	We estimate the convergence rate of the excess loss by performing linear regression between $\log(\log n/n)$ and the logarithm of the excess loss.
	Our oracle inequality suggests that the estimated coefficient should be 1.
	Indeed, we estimate the slope to be 1.03 (Figure~\ref{fig:conv_obs} a, left).
	Similarly, we can check the convergence rate of the sparse group lasso penalty of the irrelevant weights.
	Plotting the sparse group lasso penalty of the irrelevant weights against $\sqrt{\log n/n}$, we see that a near-linear relationship only begins when there are at least 400 observations (Figure~\ref{fig:conv_obs} a, right).
	Therefore we fit a linear model between $\log(\log n/n)$ and the logarithm of the weights of the irrelevant weights for the simulations with at least 400 observations.
	Our oracle inequality suggests that the estimated coefficient should be 0.5 and we estimate that the slope is 0.58, which is slightly faster than our theoretical bound.
	
	Next we check that convergence rate with respect to the number of irrelevant inputs.
	As before, the fitted networks have the same structure are the optimal network.
	Here the number of training observations remains constant ($n = 200$).
	To understand how the excess loss grows with the number of features, we perform linear regression of the excess loss against the number of covariates $p$ and its logarithm.
	According to our theoretical results, the excess loss should grow with the logarithm of $p$, not $p$ itself.
	We estimate a coefficient of 0.0003 for $p$ and a coefficient of 0.031 for $\log p$ (Figure~\ref{fig:conv_obs} b, left).
	Similarly, we check how the weight of the irrelevant weights grow with the number of features, we regress the weight against the number of covariates $p$ and the square root of the logarithm of $p$.
	We estimate a coefficient of -0.0007 for $p$ and a coefficient of 0.190 for $\sqrt{\log p}$ (Figure~\ref{fig:conv_obs} b, right).
	From these results, we conclude that our convergence rates with respect to the number of features and the number of observations appears tight.
	
	Finally, we would like to understand how the number of hidden nodes affects estimation error and discrimination of the relevant vs. irrelevant nodes.
	Recall that our derived rates depend on large powers of $m_1$, the number of hidden nodes in the first layer; however we suspect that our upper bounds are not tight with respect to $m_1$.
	We keep the number of observations and features constant ($n = 200, p = 50$), but grow the number of nodes in the single hidden layer.
	Surprisingly, we find that the excess loss is relatively constant as the number of hidden nodes grows (Figure~\ref{fig:conv_obs} c, left).
	Likewise, we find that the weight of the irrelevant nodes does not grow with the number of hidden nodes (Figure~\ref{fig:conv_obs} c, right).
	The most likely explanation is that the lasso encourages sparsity in the number of hidden nodes in the first hidden layer since it encourages weights to be set to zero.
	In the empirical results in the following sections, we find that SPINN can often zero out a large proportion of hidden nodes.

	\begin{figure}
		\caption{Simulation results investigating convergence rates and effects of penalty parameters}
		\label{fig:conv_obs}
		\vspace{-0.3in}
		\begin{center}
		\begin{tabular}{ll}
			\multicolumn{2}{p{\textwidth}}{(a) Effects of varying number of observations ($n= 100, 200, ... , 3200$). 
			}\\
			\includegraphics[width=0.35\textwidth]{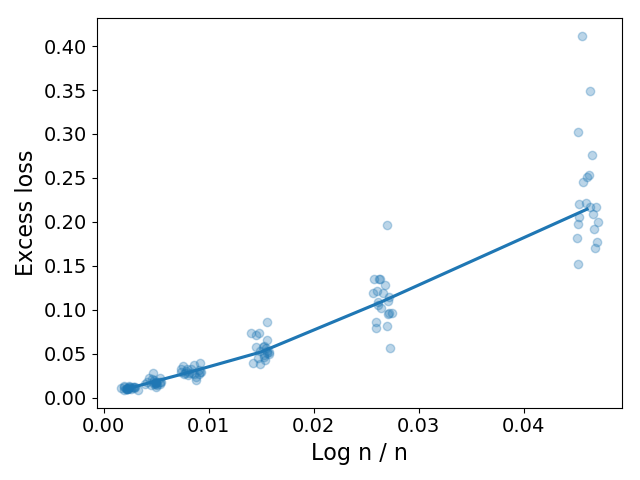}
			&
			\includegraphics[width=0.35\textwidth]{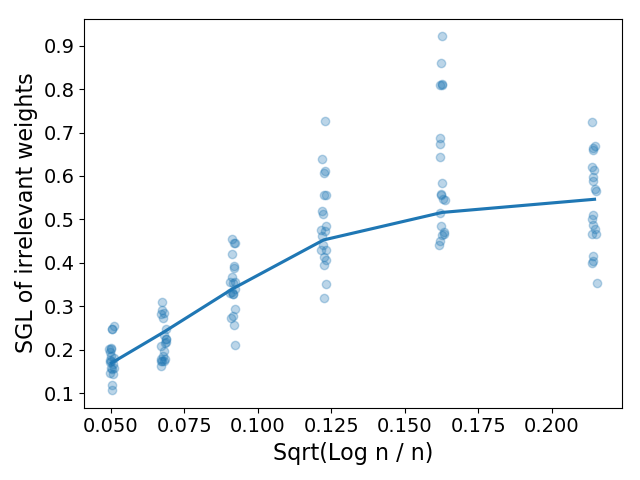}
			\\
			\multicolumn{2}{p{\textwidth}}{(b) Effects of varying number of covariates ($p= 25, 50, ..., 400$).
			}\\
			\includegraphics[width=0.35\textwidth]{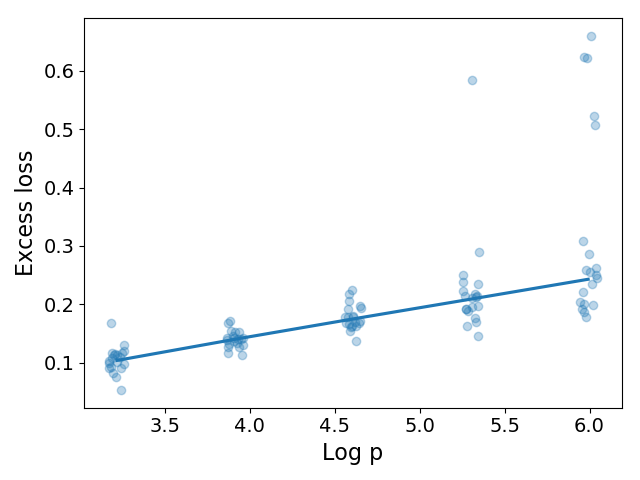}
			&
			\includegraphics[width=0.35\textwidth]{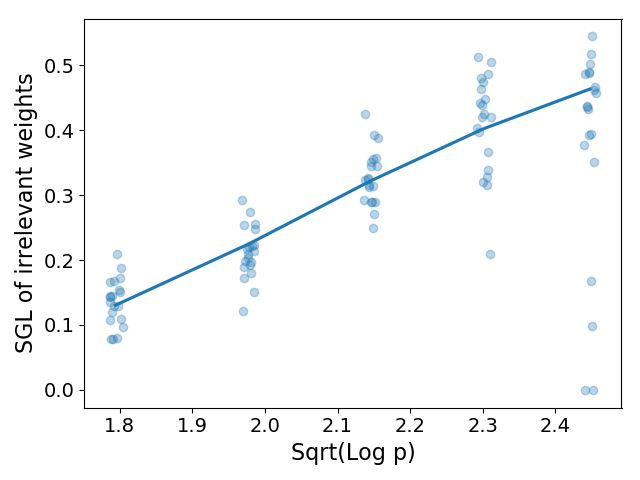}
			\\
			\multicolumn{2}{p{\textwidth}}{(c)  Effects of varying number of hidden nodes ($m_1 = 4, 8, 12, 16$).}\\
			\includegraphics[width=0.35\textwidth]{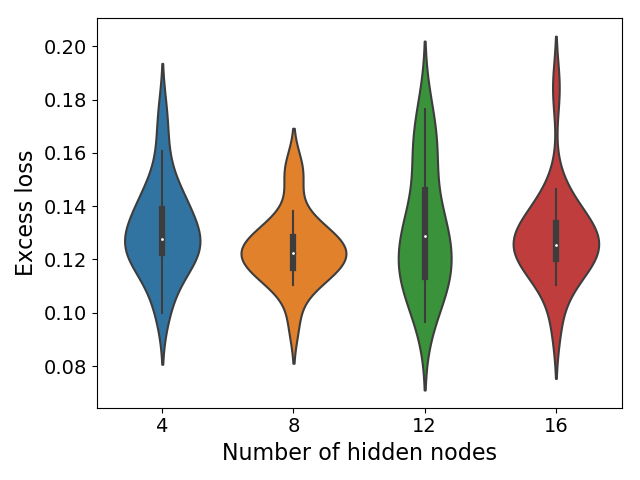}
			&
			\includegraphics[width=0.35\textwidth]{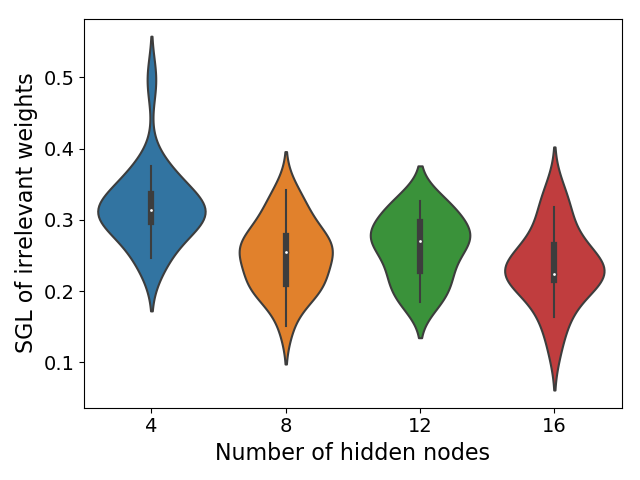}
			\\
			\multicolumn{2}{p{\textwidth}}{
				(d) Effects of varying the lasso and group lasso penalty parameters. The white box in the right and left of each plot corresponds to the penalty parameter setting that has the lowest mean squared error when upweighting the group lasso and the lasso, respectively.
				The heatmaps plot the mean squared error (left), the proportion weight assigned to relevant nodes (middle), and the proportion weight assigned to the irrelevant nodes (right).}\\
			\multicolumn{2}{c}{
				\includegraphics[height=0.29\textwidth]{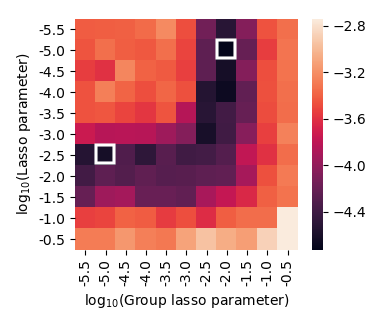}
				\includegraphics[height=0.29\textwidth]{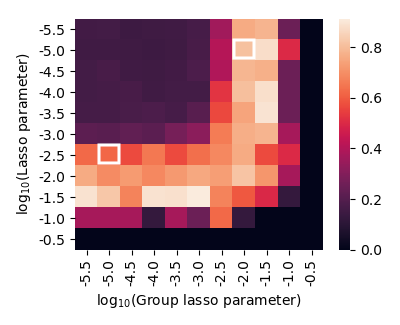}
				\includegraphics[height=0.3\textwidth]{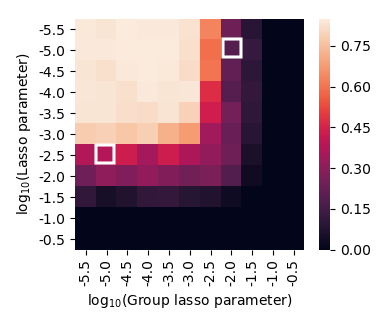}
			}
		\end{tabular}
	\end{center}
	\end{figure}
	
	\subsection{The role of lasso vs. group lasso}
	The $\alpha$ parameter in the sparse group lasso balances the shrinkage behavior of the lasso and the group lasso.
	To understand the effects of the lasso and group lasso on the behavior of SPINN, we performed a simulation study where we varied their penalty parameters (Figure~\ref{fig:conv_obs}d).
	Here the true function is
	\begin{align*}
	f^*(x) & = \sin(x_1 (x_1 + x_2)) \cos(x_3 + x_4 x_5)\sin(e^{x_5} + e^{x_6} - x_2).
	\end{align*}
	We generate $50$ covariates (where only the first six are relevant), and 250 observations.
	
	There is a near symmetry in using the lasso and the group lasso penalties.
	The mean squared error (MSE) is small when either the lasso is large and the group lasso is small or the lasso is small and the group lasso is large.
	For both scenarios, the weights tied to the relevant inputs are large and those tied to the irrelevant ones are small.
	The MSE is large when the lasso and the group lasso are both large or are both small; this is unsurprising since the former means that we are over-penalizing and the latter means that we are under-penalizing.
	
	Nonetheless there is still a difference between the two penalties: upweighting the lasso leads to a fitted model with a larger proportion of the weights on the irrelevant nodes and a smaller proportion of the weights on the relevant nodes.
	Since the group lasso penalty is better able to discriminate between inputs, upweighting the group lasso results in models with slightly smaller mean squared error (MSE = 0.009 vs. 0.010).
	
	\subsection{Comparing against other methods}
	We now present a simulation study to understand how sparse-input neural networks compare against other methods.
	We consider scenarios where the true function is the sum of univariate functions, a complex multivariate function, and a function that is in between these two extremes.
	In all cases, the true function is sparse and only depends on the first six variables.
	We compare sparse-input neural networks to the following methods:
	\begin{itemize}
		\item neural network with a single hidden layer and a ridge penalty on all the weights (ridge-only neural network);
		\item Sparse Additive Model (SpAM), which fits an additive univariate model with a sparsity-inducing penalty \citep{ravikumar2007spam}.
		\item random forests \citep{Breiman2001-fc}
		\item oracle ridge-only neural network involving only the first six covariates;
		\item oracle additive univariate model of the form
		$
		\sum_{i=1}^6 g_i (x_i)
		$
		where $g_i$ are fit using additive smoothing splines;
		\item oracle general multivariate model $g (x_1, ..., x_6)$ where $g$ is fit using a 6-variate smoothing spline.
	\end{itemize}
	The last three methods are oracles since they take only the relevant covariates as input.
	These oracle methods are not competitors in practice since they use information which will not be available; however, they give us some idea of how well our feature selection is working. Performance of the methods is assessed by the mean squared error, evaluated over 2000 points drawn from $\mathbb{P}_{X}$.
	
	We generated a training set of varying sizes and a test set of 2000 observations.
	The penalty parameters in all the methods were tuned using 3-fold cross validation.
	When tuning the hyper-parameters in SPINN, we considered up to three hidden layers and up to fifteen hidden nodes per layer.
	We use $\tanh$ as the activation function in our neural networks.
	Each simulation was repeated 20 times.
	
	\subsubsection{Additive univariate model}
	\label{sec:additive}
	In the first scenario, we have $p = 50$ covariates and the true function is the sum of univariate functions
	$$
	f^*(x)= \sin(2 x_1) + \cos(5 x_2) + x_3^3 - \sin(x_4) + x_5 - x_6^2.
	$$
	Since the true model is additive, we expect that the additive univariate oracle performs the best, followed by SpAM. As shown in Figure~\ref{fig:simulations} a, we see that this is indeed the case.
	
	We find that sparse-input neural networks also perform quite well. Thus if we are unsure if the true function is the sum of univariate functions, a sparse-input neural network can be a good option. In addition, we notice that the performance of sparse-input neural networks tends to track the oracle neural network and the multivariate oracle. In small sample sizes, sparse-input neural networks and oracle neural networks perform better than the multivariate oracle, as there is not enough data to support fitting a completely unstructured $6$-variate smoother. As the number of samples increase, the multivariate oracle overtakes the sparse-input neural network since it knows which features are truly relevant. We observe similar trends in the next two scenarios.
	The ridge-only neural network and random forests perform poorly in this scenario since they are unable to determine which variables are relevant.
	
	Since the true model is an additive univariate model, a neural network that closely approximates $f^*$ should capture this additive structure by estimating a separate network for each covariate and then output a linear combination of their outputs.
	As such, the nodes in the first hidden layer of this ``additive'' neural network should be very sparsely connected: each of these nodes should be connected to exactly a single input node.
	We were not able to recover this behavior in the fitted networks using SPINN due to the small sample size and because neural networks tend to prefer fitting models with interactions rather than univariate additive models.
	Nonetheless, the lasso did seem to encourage fitting models where the hidden nodes were only connected to a small number of input nodes.
	For example, in a simulation replicate with 2000 training observations where 43 features were included in the fitted model, some hidden nodes had 36 nonzero incoming edges whereas others had 43 nonzero incoming edges.
	
	\begin{figure}
		\begin{tabular}{ccc}
			\includegraphics[width=0.35\textwidth]{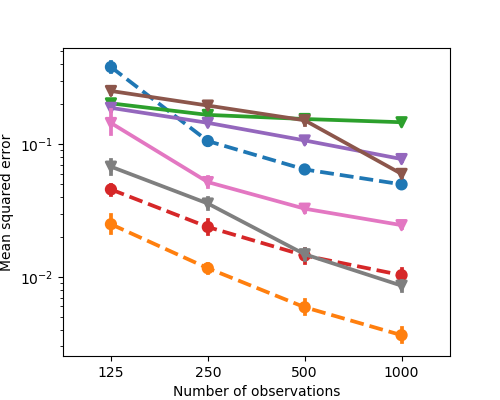}&
			\includegraphics[width=0.35\textwidth]{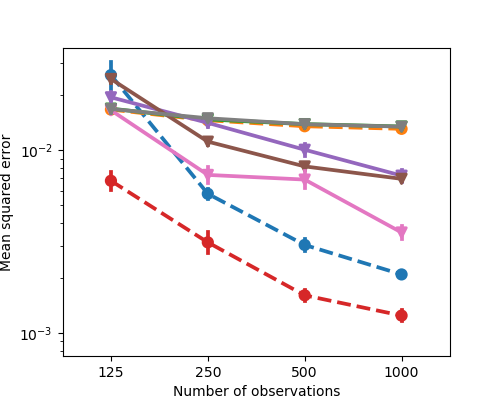}&
			\includegraphics[width=0.35\textwidth]{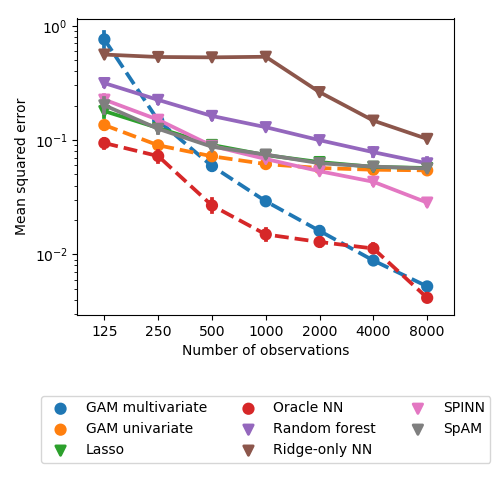}\\
			(a) Additive Univariate & (b) Complex Multivariate & (c) High-dimensional
	\end{tabular}
		\caption{Simulation results from the three scenarios: additive univariate function with $p = 50$ (a), a complex multivariate function $p = 50$ (b), and the sum of multivariate functions with $p= 1000$ (c). The dashed lines and triangles indicate oracle models.}
		\label{fig:simulations}
	\end{figure}
	
	\subsubsection{Complex multivariate model}
	
	In the second simulation, we use $p = 50$ covariates and a sparse, multivariate generative function
	\begin{align*}
	f^*(x) & = \sin(x_1 (x_1 + x_2)) \cos(x_3 + x_4 x_5)\sin(e^{x_5} + e^{x_6} - x_2).
	\end{align*}
	Here we expect the general multivariate oracle to perform the best in large sample sizes, which is confirmed in Figure~\ref{fig:simulations} b. Similar to results in Section~\ref{sec:additive}, the performance of sparse-input neural networks more closely tracks the trajectory of oracle neural networks and the general multivariate oracle. As expected, the additive univariate oracle and SpAM perform very poorly. Their MSEs flatten out very quickly due to the bias from assuming an additive univariate model. In fact, we see that given a sufficiently large training set, even the ridge-only neural network and random forests can outperform the additive univariate methods.
	
	\subsubsection{High-dimensional additive multivariate model}
	
	Finally we consider a setting where we have a large number of input features, $p = 1000$. We use a regression function that is the sum of 3- and 4-variate functions:
	\begin{align*}
	f^*(x) & = (x_1 \wedge x_2) \cos(1.5 x_3 + 2 x_4) + e^{x_5 + \sin(x_4)} x_2 + \sin(x_6 \vee x_3)) (x_5 - x_1).
	\end{align*}
	This places it between the simple additive univariate function in the first scenario and the complex 6-variate function in the second scenario.
	
	The results (Figure~\ref{fig:simulations} c) are a mixture of the results from the previous two simulation setups: SpAM and sparse-input neural networks perform similarly in small samples and diverge at larger sample sizes; and random forests and ridge-penalized neural networks perform poorly.
	These results illustrate that the utility of SPINN depends on the problem.
	Since $f^*$ is the sum of low-dimensional functions, models that ignore interactions can still perform well when the number of observations is much smaller than the number of features; SPINN is more useful in this problem as it becomes more moderate-dimensional.
	
	\section{Data analyses}
	\label{sec:data}
	
	We now compare SPINN to the other methods on several high- and moderate-dimensional datasets.
	The first task is to predict phenotypes given gene expression microarray data.
	The second task is to predict binding affinity between peptides and protein receptors.
	Employing nonparametric techniques in these problems may improve performance since the biological processes are typically believed to involve complex higher-order interactions between various components.
	
	\subsection{Predicting phenotypes from gene expression data}
	Here we consider two very high-dimensional problems involving gene expression data.
	The first dataset is a regression problem where our goal is to predict the production rate of riboflavin in \textit{Bacillus subtilis} given gene expression profiles of $p=4088$ genes in $n = 71$ experimental settings \citep{buhlmann2014high}.
	The second dataset entails predicting the molecular type of samples (NEG or not NEG) collected from $n = 128$ patients with acute lymphoblastic leukemia \citep{Chiaretti:2004gq, ALL}.
	For each subject, $p = 12625$ gene expression levels were measured using Affymetrix human 95Av2 arrays and normalized using the robust multichip average method.
	Based on pre-tuning the network structure, we cross-validated over networks with one or two hidden layers where the first layer had ten nodes and the second layer had up to four nodes for the first dataset.
	For the second dataset, we only considered a single hidden layer where the first layer had three or ten nodes.
	
	For both datasets, we randomly split the dataset such that test set makes up one-fifth of the data and the rest are for training.
	We compare the methods based on the average performance across 30 random splits of the dataset (Table~\ref{table:ribo}).
	We also evaluate how many features are included in the fitted model for the various methods.
	For a neural network, we consider a feature to be included in the fitted model if the norm of the weights connected that that feature is nonzero.
	
	The top performing methods on these two datasets are sparse-input neural networks and simple linear/logistic models employing the lasso penalty.
	SPINN achieves the best estimated error in the first dataset and the two methods are not significantly different in the second.
	This is particularly impressive as these datasets are extremely high-dimensional and ``common knowledge'' has strongly pushed against using neural networks in such settings, as evidenced by the poor performance of the standard ridge-penalized network.
	These experiments suggest that SPINN is promising for analyzing gene expression datasets, particularly with the recent push to generate larger genomic datasets \citep{Collins2015-kz}.
	
	{\setstretch{1.0}
	\begin{table}
		\caption{
			\label{table:ribo}
			Average performance of the different methods for predicting riboflavin production rates in \textit{Bacillus subtilis} (top) and classification of Leukemia samples (bottom).
			The ``\# feats incl.'' indicates the number of features included in the fitted model.
			Standard error given in parentheses.
		}
		\centering
		\begin{tabular}{lll}
			\multicolumn{3}{c}{\textit{Predicting riboflavin production rates}}\\
			Method & Mean squared error & \# feats incl. \\ \midrule
			Sparse-input NN & 0.116 (0.009)  & 45.1 (1.6) \\
			Ridge-penalized NN  & 0.145 (0.014) & 4088 (0) \\
			SpAM             & 0.148 (0.015)   & 32.9 (2.1) \\
			Linear model with Lasso & 0.128 (0.009) & 19.9 (1.26) \\
			Random forest & 0.206 (0.020) & --\\
			\toprule
			\multicolumn{3}{c}{\textit{Classification of Leukemia samples}}\\
			Method & Classification error & \# feats incl. \\ \midrule
			Sparse-input NN & 0.149 (0.012) & 43.7 (2.1) \\
			Ridge-penalized NN  & 0.183 (0.015) & 12625 (0) \\
			SpAM             & 0.231 (0.015) & 47.1 (4.2)  \\
			Logistic regression with Lasso & 0.141 (0.011) & 44.1 (2.2) \\
			Random forest & 0.179 (0.011) & --
		\end{tabular}
	\end{table}
}
	
	\subsection{Peptide-MHC binding affinity prediction}
	
	Here we apply our method to predict the binding affinity between peptides and class I major histocompatibility complexes (MHCs).
	Class I MHCs are receptors on the surface of nearly every nucleated cell in the body and bind to peptides (typically eight to eleven amino acids) derived from intracellular proteins.
	These peptide-MHC complexes help the immune system monitor the health of each cell and eliminate infected or tumoral cells.
	Peptides that cannot be bound by an individual's MHC are ``invisible'' to the T cell component of their immune system, thus understanding the binding affinity between the peptide and MHC is important developing vaccines and immunotherapies against cancer.
	
	Each MHC molecule can only bind to a specific subset of peptides and the binding affinity depends on how amino acids interact.
	As such, the most accurate algorithms for predicting binding affinity are fully nonparametric \citep{Andreatta2016-bk, Jurtz2017-xm, ODonnell2018-pn, Boehm2018-di} and the start-of-the-art method employs ridge-penalized neural networks \citep{Trolle2015-lc, Zhao2018-za}.
	Here we apply our method to see if incorporating the sparse group lasso can further improve prediction accuracy.
	
	Since most peptides that bind to MHC are composed of nine amino acids (AAs), we will only consider peptides of this length in this analysis.
	We used an ``allele-specific'' approach where separate predictors are trained for each MHC allele.
	We trained the models on quantitative binding affinity data available from \url{https://github.com/openvax/mhcflurry}, which was collated from the Immune Epitope database (IEDB) \citep{Vita2015-ka} and a published benchmark \citep{Kim2014-iq}.
	For each amino acid, we generated a one-hot vector as well as the following physical features: Kidera factors \citep{Kidera1985-wg} and features generated in ForestMHC \citep{Boehm2018-di} (hydropathy score, molar mass, and whether or not the amino acid was aromatic).
	In total, each input amino acid was represented by a 33-dimensional vector and each peptide was represented by 297 covariates.
	As there are many MHC alleles, we sampled three alleles to benchmark performance: HLA-A*01:01, HLA-B*44:02, and HLA-B*08:02.
	These were chosen to represent datasets with different number of observations.
	We considered neural network architectures with one or two hidden layers, where the first layer had 45 nodes and the second layer had 5 to 20 nodes.
	We split the dataset into a training and test set in the same manner as before.
	
	SPINN performed significantly better for two of the three HLA alleles and tied with random forests for allele HLA-B*44:02 (Table~\ref{table:binding}).
	Since the peptide-MHC binding affinity depends on interactions between amino acids, it is unsurprising that Lasso and SpAM performed poorly.
	Given that there was a significant improvement over the ridge-penalized neural network, we believe that incorporating a sparse group lasso penalty can significantly improve existing peptide-MHC binding affinity prediction methods.
	
	For each HLA allele, we found that the estimated SPINN model used at least one generated feature from each position.
	For each position in the peptide, the nonzero features typically included a subset of the Kidera factors, some elements from the one-hot vector, and a subset of the remaining features from ForestMHC.
	By modeling interactions between input features, SPINN tended to employ more input features than the additive models SpAM and Lasso.
	Nonetheless, SPINN zeroed-out many features; even for HLA-A*01:01 which had the most observations, the fitted model only used an average of 154 of the 297 features.
	
	{\setstretch{1.0}
	\begin{table}
		\caption{
			\label{table:binding}
			Average performance of peptide-MHC binding prediction for three HLA alleles over 30 replicates.
			The ``\# feats'' indicates the number of features included in the fitted model.
			MSE = mean squared error.
			Standard error given in parentheses.
		}
		\centering
		\hspace{-.5in}
		\begin{tabular}{l|cc|cc|cc}
			& \multicolumn{2}{c}{HLA-A*01:01}& \multicolumn{2}{|c}{HLA-B*44:02} & \multicolumn{2}{|c}{HLA-B*08:02} \\
			& \multicolumn{2}{c}{$n = 3370$}& \multicolumn{2}{|c}{$n = 1430$} & \multicolumn{2}{|c}{$n = 487$} \\
			Method & MSE &  \# feats & MSE & \# feats & MSE & \# feats \\ \midrule
			SPINN & \textbf{0.147} (0.002) & 154 (2)& \textbf{0.193} (0.007) & 116 (3) & \textbf{0.17} (0.01) & 97 (3) \\
			Ridge NN & 0.249 (0.022) & 297 (0)  & 0.268 (0.022) & 297 (0) & 0.25 (0.03)  & 297 (0) \\
			SpAM         & 0.247 (0.003) & 63 (1)  & 0.262 (0.012) & 59 (2) & 0.24 (0.01) & 33 (2) \\
			Lasso &  0.247 (0.003) & 107 (2) & 0.257 (0.006) & 96 (4) & 0.26 (0.01) & 67 (2) \\
			Random forest & 0.161 (0.003) & -- & \textbf{0.191} (0.005) & -- & 0.21 (0.01) & --
		\end{tabular}
	\end{table}
}
	
	\section{Discussion}
	
	We have introduced the use of sparse-input neural networks as a nonparametric regression and classification method for high-dimensional data, where the first-layer weights are penalized with a sparse group lasso.
	When the true model is best approximated by a sparse network, we show that the fitted model using the sparse group lasso has a prediction error that grows logarithmically with the number of features.
	Thus sparse-input neural networks can be effective for modeling high-dimensional data.
	We have also provided empirical evidence via simulation studies and data analyses to show that sparse-input neural networks can outmatch other more traditional nonparametric methods in high dimensions.
	Our results show that neural networks can indeed be applied to high-dimensional datasets, as long as proper regularization is applied.
	
	One possible direction of further research is to understand how sparsity-inducing penalties in upper-layer weights can be used to tune the size of the network and how it affects prediction accuracy.
	If the size of such networks can be controlled by a small number of penalty parameters, then this could potentially reduce the number of hyper-parameters that need to be tuned.
	
	One drawback of sparse-input neural networks, and neural networks in general, is that they require a significant amount of time to train.
	Much of the training time is spent on tuning the hyper-parameters and testing different initializations since the objective function is non-convex.
	On the other hand, other methods like sparse additive models and random forests are much faster to fit.
	The advantage of SPINN over additive models is most apparent when modeling complex higher-order interactions drastically improves performance; the advantage of SPINN over other fully nonparametric methods is most apparent when the dataset is very high-dimensional but the number of observations is relatively small.
	
	The code for fitting sparse-input neural networks is available at \url{http://github.com/jjfeng/spinn}.
	
	\section*{Acknowledgments}
	We thank Frederick A. Matsen IV for helpful discussion and suggestions.

\pagebreak
\appendix
\section{Proofs}

\subsection{Proof of Theorem~\ref{thm:sparse_nn}}
The proof for Theorem~\ref{thm:sparse_nn} is composed of two main steps. First, we show that the excess loss of any $f_{\boldsymbol{\eta}}$ is lower bounded by a quadratic function of the distance from $\boldsymbol{\eta}$ to $EQ^*$.
We combine this with the definition of $\hat{\boldsymbol{\eta}}$ and $\mathcal{T}_{\tilde{\lambda}, T}$ to derive the result. 

Using Conditions~\ref{cond:hess}, \ref{cond:identifiable}, and \ref{condn:bound_third}, we lower bound the excess loss by a quadratic function.
Let $m^*$ be the number of parameters in the network excluding the ones connected to the irrelevant weights.
\begin{lemma}[Quadratic lower bound]
	\label{lemma:quad_margin}
	Suppose Conditions~\ref{cond:hess}, \ref{cond:identifiable}, \ref{condn:bound_third} hold.
	For some constant $R > 0$, suppose $\boldsymbol{\eta} \in \Theta$ satisfies
	\begin{align}
	\| \boldsymbol{\eta}_S - \boldsymbol{\eta}^{*(\boldsymbol{\eta})} \|_2 \le R
	\label{eq:quad_cond_bound}
	\end{align}
	and 
	\begin{align}
	\|\boldsymbol \theta_{S^{c}}\|_{1} \le 
	c \left(
	\left\Vert \other(\boldsymbol \eta)- \other \left( \boldsymbol \eta^{*(\boldsymbol \eta)} \right ) \right\Vert _{2}
	+ 3 \sum_{j \in S}
	\Omega_\alpha \left ( \boldsymbol \theta_{1,\cdot, j} - \boldsymbol \theta_{1,\cdot, j}^{*(\boldsymbol \eta)} \right )
	\right)
	\label{eq:quad_cond}
	\end{align}
	for some $c > 0$.
	Then we have
	$
	\mathcal{E}\left(\boldsymbol \eta\right)\ge  \| \boldsymbol{\eta}_S - \boldsymbol{\eta}^{*(\boldsymbol{\eta})} \|_2^2/C_{0}^{2}
	$
	where
	\begin{align}
	C_{0}^{2} & = \frac{1}{\epsilon_{0}} \vee \frac{R^{2}}{\chi_{\epsilon_{0}}},\\
	\epsilon_{0} & =\frac{h_{min}}{C_1 G c^3 \left(
		(1 - \alpha) \sqrt{m_1|S|}  + \alpha \sqrt{|S|}
		+ \sqrt{m^*}
		\right)^3}
	\end{align}
	for some constant $C_1 > 0$.
\end{lemma}

\begin{proof}
	For convenience, denote $ d(\boldsymbol{\eta}) = \| \boldsymbol{\eta}_S - \boldsymbol{\eta}^{*(\boldsymbol{\eta})} \|_2$.
	Since the gradient of the loss function is zero at all $\boldsymbol{\eta}^* \in EQ^*$, we have by Taylor expansion
	\begin{align}
	\mathcal{E}\left(\boldsymbol \eta\right)=
	\frac{1}{2}
	\left(\boldsymbol \eta-\boldsymbol \eta^{*(\boldsymbol \eta)}\right)^{\top}
	\left[\nabla_{\eta}^{2}\mathbb{P}\ell_{\eta}(y, \boldsymbol x)\right]_{\eta=\eta^{*(\boldsymbol \eta)}}
	\left(\boldsymbol \eta-\boldsymbol \eta^{*(\boldsymbol \eta)}\right)+r_{\eta}\label{eq:taylor_exp}
	\end{align}
	where $r_\eta$ is the Lagrange remainder.
	By Condition~\ref{cond:hess}, we have that
	\begin{align}
	\left(\boldsymbol \eta-\boldsymbol \eta^{*(\boldsymbol \eta)}\right)^{\top}
	\left[\nabla_{\eta}^{2}\mathbb{P}\ell_{\eta}(y, \boldsymbol x)\right]_{\eta=\eta^{*(\boldsymbol \eta)}}
	\left(\boldsymbol \eta-\boldsymbol \eta^{*(\boldsymbol \eta)}\right)
	\ge h_{min}d^2(\boldsymbol{\eta}).
	\label{eq:hess_implication}
	\end{align}
	In addition, $r_\eta$ is bounded above by
	\begin{align}
	|r_\eta| & \le 
	\frac{1}{6}\|\boldsymbol \eta-\boldsymbol \eta^{*(\boldsymbol \eta)}\|_{1}^{3} 
	\mathbb{P}
	\left [
	\sup_{\eta\in\Theta} \max_{j_1, j_2, j_3} 
	\left |
	\frac{\delta^3 \ell_{\eta}(y, \boldsymbol x) }{\delta\eta_{j_1} \delta\eta_{j_2} \delta\eta_{j_3}}
	\right |
	\right ]\\
	& \le \frac{G}{6}\|\boldsymbol \eta-\boldsymbol \eta^{*(\boldsymbol \eta)}\|_{1}^{3}
	\label{eq:eta_minus_eta}
	\end{align}
	where we used Condition~\ref{condn:bound_third} in the last line.
	Moreover, if \eqref{eq:quad_cond} holds, then
	\begin{align}
	\|\boldsymbol \theta_{S^c}\|_{1}
	& \le
	c\left(
	\left\Vert \other(\boldsymbol \eta)- \other \left( \boldsymbol \eta^{*(\boldsymbol \eta)} \right ) \right\Vert _{2}
	+ 3 \left [ (1 - \alpha) \sqrt{m_1} + \alpha \right ] \sqrt{|S|} 
	\left \|\boldsymbol \theta_{1,\cdot,S}  - \boldsymbol \theta_{1,\cdot, S}^{*(\boldsymbol \eta)} \right \|_2
	\right )\\
	& \le c \left (1 + 3 (1 - \alpha) \sqrt{m_1|S|}  + 3 \alpha \sqrt{|S|} \right ) d(\boldsymbol{\eta}).
	\label{eq:eta_extra}
	\end{align}
	Since $
	\|\boldsymbol{\eta} - \boldsymbol{\eta}^{*(\boldsymbol{\eta})}\|_1
	= \|\boldsymbol{\theta}_{1,\cdot,S^c}\|_1
	+ \|\boldsymbol{\eta}_S - \boldsymbol{\eta}^{*(\boldsymbol{\eta})}\|_1$,
	we can combine \eqref{eq:eta_minus_eta} and \eqref{eq:eta_extra} to get
	\begin{align}
	|r_\eta| & \le \frac{G}{6} c^3 \left(
	1 + 3 (1 - \alpha) \sqrt{m_1|S|}  + 3 \alpha \sqrt{|S|}
	+ \sqrt{m^*}
	\right)^3
	d^3(\boldsymbol{\eta})\\
	& \le C_1 G c^3 \left(
	(1 - \alpha) \sqrt{m_1|S|}  + \alpha \sqrt{|S|}
	+ \sqrt{m^*}
	\right)^3
	d^3(\boldsymbol{\eta})
	\label{eq:lagrange}
	\end{align}
	for a constant $C_1$.
	Combining \eqref{eq:hess_implication} and \eqref{eq:lagrange} gives us
	\[
	\mathcal{E}(\boldsymbol \eta ) \ge
	\frac{1}{2}h_{min}d^2(\boldsymbol \eta)- 
	\frac{G}{6}
	c^3
	\left(
	(1 - \alpha) \sqrt{m_1|S|}  + \alpha \sqrt{|S|}
	+ \sqrt{m^*}
	\right)^3
	d^3(\boldsymbol \eta).
	\]
	Now we use Condition~\ref{cond:identifiable} and apply Auxiliary Lemma in \citet{stadler2010}.
	In particular, the Auxiliary Lemma states that
	\begin{quote}
		\textbf{Auxiliary Lemma in \citet{stadler2010}}
		Let $h: [-R, R] \mapsto [0, \infty)$ have the following properties:
		\begin{enumerate}
			\item $\forall \epsilon > 0, \exists \chi_\epsilon > 0$ such that $\inf_{\epsilon < z\le R} h(z)\ge \chi_{\epsilon}$,
			\item $\exists \Lambda > 0, C> 0$, such that $\forall z \le R, h(z) \ge  \Lambda^2 z^2 - Cz^3$.
		\end{enumerate}
		Then $\forall |z| \le R$,
		\begin{align}
		h(z) \ge z^2 / C_0^2
		\end{align}
		where
		\begin{align}
		C_0^2 = \frac{1}{\epsilon_0} \vee \frac{K_0^2}{\chi_{\epsilon_{0}}},
		\epsilon_{0} = \frac{\Lambda^2}{2C}.
		\end{align}
	\end{quote}
	To use the Auxiliary Lemma, plug in $z = d(\boldsymbol{\eta})$.
\end{proof}

Finally, we are ready to prove Theorem~\ref{thm:sparse_nn}.
In this document, we use the notation $\ell_{\eta}(y,\boldsymbol x) = \ell(y,f_{\eta}(\boldsymbol x))$.
\begin{proof}[Proof of Theorem~\ref{thm:sparse_nn}]
	For simplicity, we'll denote $\boldsymbol \eta^{*(\hat\eta)}$, the closest point in $EQ^*$ to $\hat{\boldsymbol{\eta}}$,  by $\boldsymbol \eta^*$.
	By definition of $\hat {\boldsymbol{\eta}}$,
	$$
	\mathbb{P}_{n}\ell_{\hat{\eta}}(y,\boldsymbol x)
	+ \lambda \sum_{j=1}^p \Omega_\alpha (\hat{\boldsymbol \theta}_{1,\cdot,j} )
	\le\mathbb{P}_{n}\ell_{\eta^*}(y,\boldsymbol x)
	+\lambda \sum_{j\in S} \Omega_\alpha (\boldsymbol \theta_{1,\cdot,j}^{*}).
	$$
	Rearranging, we get
	\begin{align*}
	\mathcal{E}\left(\hat{\boldsymbol\eta}\right)+\lambda \sum_{j=1}^p  \Omega_\alpha (\hat{\boldsymbol\theta}_{1,\cdot,j} )
	& \le \left|\left(\mathbb{P}_{n}-\mathbb{P}\right)\left(\ell_{\eta^*}(y,\boldsymbol x)-\ell_{\hat{\eta}}(y,\boldsymbol x)\right)\right|+
	\lambda \sum_{j\in S} \Omega_\alpha (\boldsymbol \theta^*_{1,\cdot,j}).
	\end{align*}
	Over the set $\mathcal{T}_{\tilde{\lambda}, T}$, we have that
	\begin{align}
	\begin{split}
	\label{eq:basic_ineq}
	\mathcal{E}\left(\hat{\boldsymbol \eta}\right)+\lambda \sum_{j=1}^p  \Omega_\alpha (\hat{\boldsymbol\theta}_{1,\cdot,j} )
	& \le T\tilde{\lambda}\left(
	\tilde{\lambda}
	\vee
	\left (
	\left\Vert \other(\hat{\boldsymbol \eta})-\other(\boldsymbol{\eta}^*) \right\Vert _{2}+
	\sum_{j=1}^p \Omega_\alpha (\hat{\boldsymbol \theta}_{1,\cdot,j} -\boldsymbol \theta_{1,\cdot,j}^*)
	\right)
	\right)\\
	& \quad + \lambda \sum_{j\in S} \Omega_\alpha (\boldsymbol \theta_{1,\cdot,j}^*).
	\end{split}
	\end{align}
	Now we consider three possible cases.
	\paragraph{Case 1:}
	$
	\tilde{\lambda}\ge
	\left\Vert \other(\hat{\boldsymbol \eta})-\other(\boldsymbol{\eta}^*) \right\Vert _{2}+
	\sum_{j=1}^p \Omega_\alpha (\hat{\boldsymbol \theta}_{1,\cdot,j} -\boldsymbol \theta_{1,\cdot,j}^*)
	$
	
	In this case, we can rearrange \eqref{eq:basic_ineq} and apply the triangle inequality to get
	\begin{align}
	\mathcal{E}\left(\hat{\boldsymbol \eta}\right)+\lambda \sum_{j \in S^c} \Omega_\alpha (\hat{\boldsymbol\theta}_{1,\cdot,j} )
	& \le T\tilde{\lambda}^{2}+\lambda \sum_{j \in S} \Omega_\alpha (\hat{\boldsymbol\theta}_{1,\cdot,j} - \boldsymbol \theta_{1,\cdot,j}^*).
	\end{align}
	Since $\tilde{\lambda} \ge \sum_{j=\in S} \Omega_\alpha (\hat{\boldsymbol\theta}_{1,\cdot,j} - \boldsymbol \theta_{1,\cdot,j}^*)$, we have
	\begin{align}
	\mathcal{E}\left(\hat{\boldsymbol \eta}\right)+\lambda \sum_{j \in S^c} \Omega_\alpha (\hat{\boldsymbol\theta}_{(j)} )
	& \le T\tilde{\lambda}^{2}+\lambda\tilde{\lambda}.
	\end{align}
	
	\paragraph{Case 2:}
	$
	\tilde{\lambda}<
	\left\Vert \other(\hat{\boldsymbol \eta})-\other(\boldsymbol{\eta}^*) \right\Vert _{2}+
	\sum_{j=1}^p \Omega_\alpha (\hat{\boldsymbol \theta}_{1,\cdot,j} -\boldsymbol \theta_{1,\cdot,j}^*)
	$
	
	In this case, we can rearrange \eqref{eq:basic_ineq} and apply the triangle inequality to get
	\begin{align}
	& \mathcal{E}\left(\hat{\boldsymbol \eta}\right)+
	(\lambda - T\tilde{\lambda})
	\sum_{j \in S^c} \Omega_\alpha(\hat{\boldsymbol \theta}_{1,\cdot,j})\\
	& \le T\tilde{\lambda}
	\left\Vert \other(\hat{\boldsymbol \eta})- \other (\boldsymbol \eta^*)\right\Vert _{2}
	+(\lambda + T\tilde{\lambda}) 
	\sum_{j \in S} \Omega_\alpha(\hat{\boldsymbol \theta}_{1,\cdot,j}-\boldsymbol \theta^*_{1,\cdot,j}).
	\label{eq:case2}
	\end{align}
	
	We now show that \eqref{eq:quad_cond} in Lemma~\ref{lemma:quad_margin} is satisfied.
	To see this, note that
	\begin{align}
	\|\hat{\boldsymbol \theta}_{S^c}\|_1 
	&\le \frac{1}{1 -\alpha + \alpha/\sqrt{m_1}} \sum_{j \in S^c} \Omega_\alpha(\hat{\boldsymbol \theta}_{(j)}).
	\end{align}
	Also, we have from \eqref{eq:case2}
	\begin{align}
	\sum_{j \in S^c} \Omega_\alpha(\hat{\boldsymbol \theta}_{(j)})
	& \le \frac{T\tilde{\lambda}}{\lambda - T\tilde{\lambda}}
	\left\Vert \other(\hat{\boldsymbol \eta})-\other(\boldsymbol \eta^* )\right\Vert _{2}
	+ \frac{\lambda + T\tilde{\lambda}}{\lambda - T\tilde{\lambda}}
	\sum_{j \in S} \Omega_\alpha (\hat{\boldsymbol \theta}_{1,\cdot,j}-\boldsymbol \theta_{1,\cdot,j})\\
	& \le 
	\left\Vert \other(\hat{\boldsymbol \eta})-\other(\boldsymbol \eta^* )\right\Vert _{2}
	+ 3 \sum_{j \in S} \Omega_\alpha (\hat{\boldsymbol \theta}_{1,\cdot,j}-\boldsymbol \theta_{1,\cdot,j}).
	\end{align}
	
	Using a similar argument, we can also show that 
	$$
	\| \boldsymbol{\eta}_S - \boldsymbol{\eta}^{*(\boldsymbol{\eta})} \|_2
	\le
	C_1 (K + \max_{\eta^* \in EQ^*} \sum_{j \in S} \Omega_\alpha(\boldsymbol \theta_{1,\cdot,j^*}))
	$$
	for some constant $C_1 > 0$.
	So \eqref{eq:quad_cond_bound} is satisfied.
	Thus Case 2 satisfies all the conditions in Lemma~\ref{lemma:quad_margin}
	and we have
	\begin{align}
	\mathcal{E}(\hat{\boldsymbol{\eta}}) \ge \left\|\hat{\boldsymbol\theta}_{S}-\boldsymbol\theta^* \right \|_2^2/C_0^2
	\label{eq:quad_lb_apply}
	\end{align} where
	\begin{align}
	C_{0}^{2} & = \frac{1}{\epsilon_{0}} \vee \frac{C^2_1 (K + \max_{\eta^* \in EQ^*} \sum_{j \in S} \Omega_\alpha(\boldsymbol \theta_{1,\cdot,j^*}))^{2}}{\chi_{\epsilon_{0}}},\\
	\epsilon_{0} & =\frac{h_{min} (1 -\alpha + \alpha/\sqrt{m_1})^3}{C_1 G \left(
		(1 - \alpha) \sqrt{m_1|S|}  + \alpha \sqrt{|S|}
		+ \sqrt{m^*}
		\right)^3}.
	\end{align}
	
	Now from \eqref{eq:case2}, we have that
	\begin{align}
	\mathcal{E}\left(\hat{\boldsymbol\eta}\right)
	+ (\lambda - T\tilde{\lambda}) \sum_{j \in S^c} \Omega_\alpha(\hat{\boldsymbol \theta}_{1,\cdot,j})
	& \le
	\left(T \tilde{\lambda} + \lambda \right)
	\sum_{j \in S} \Omega_\alpha (\hat{\boldsymbol\theta}_{1,\cdot,j}-\boldsymbol\theta^*_{1,\cdot,j})
	+ T\tilde{\lambda} \|\other(\hat{\boldsymbol{\eta}}) - \other(\boldsymbol{\eta}^*)\|_2\\
	& \le
	\left(T \tilde{\lambda} + \lambda \right)
	\left((1-\alpha) \sqrt{m_1} + \alpha \right )
	\sqrt{|S|}
	\left\|\hat{\boldsymbol\theta}_{S}-\boldsymbol\theta^* \right \|_2\\
	& \le
	\frac{1}{2} \left[
	\left(T \tilde{\lambda} + \lambda \right)^2
	\left((1-\alpha) \sqrt{m_1} + \alpha \right )^2 |S|
	C_0^2
	+ \frac{1}{C_0^2}\left\|\hat{\boldsymbol\theta}_{S}-\boldsymbol\theta^* \right \|_2^2
	\right]
	\end{align}
	
	Finally, plugging in \eqref{eq:quad_lb_apply}, we have that
	\begin{align}
	\frac{1}{2} \mathcal{E}\left(\hat{\boldsymbol\eta}\right)
	+ (\lambda - T\tilde{\lambda}) \sum_{j \in S^c} \Omega_\alpha(\hat{\boldsymbol \theta}_{1,\cdot,j})
	\le
	\frac{1}{2} 
	\left(T \tilde{\lambda} + \lambda \right)^2
	\left((1-\alpha) \sqrt{m_1} + \alpha \right )^2 |S|
	C_0^2.
	\end{align}
\end{proof}

\subsection{Proof of Theorem~\ref{thm:emp_proc}}

Next we use empirical process techniques to prove Theorem~\ref{thm:emp_proc} where we will measure complexity of function classes using metric entropy \citep{wellner1996}. 
Let the $u$-entropy of a function class $\mathcal{G}$ with respect to the norm $\|\cdot \|$ be denoted 
$H\left(u,\mathcal{G},\|\cdot\|\right)$, which is equal to the log of its $u$-covering number $N\left(u,\mathcal{G},\|\cdot\|\right)$.

Define the empirical process term 
\begin{equation}
v_{n}(\boldsymbol\eta)=(\mathbb{P}_n-\mathbb{P})\ell_{\eta}(y,\boldsymbol x).
\label{eq:emp_proc}
\end{equation}
The basic idea of the proof is to split $v_n(\boldsymbol\eta)$ into a truncated component $v_{n}^{\text{trunc}}(\boldsymbol\eta)$ and a remainder term $v_{n} - v_{n}^{\text{trunc}}(\boldsymbol\eta)$, where
\begin{align}
v_{n}^{\text{trunc}}(\boldsymbol\eta)= (\mathbb{P}_n-\mathbb{P}) \left[
\ell_{\eta}(y,\boldsymbol x) 1\left\{ G(y)\le M_{n}\right\}
\right] 
\end{align}
for some truncation function $G(\cdot)$ and constant $M_n > 0$ that can grow with $n$.
We will later show that choosing $M_n = O_p(\sqrt{\log n})$ is appropriate for our convergence rate bounds.
We then control the truncated empirical process and the remainder separately.

The function $G(\cdot)$ is used to truncate the gradient of the loss function with respect to the value of the nodes in the first hidden layer as well as the parameters for the upper layers, i.e. those above the first hidden layer.
Let $\boldsymbol{\eta}_{\upper}$ denote the parameters for the upper neural network.
Then $f_{\boldsymbol{\eta}_{\upper}}$ will denote the output of this separate neural network.
Then define
\begin{align}
\tilde{\ell}(y, z, \boldsymbol{\eta}_{\upper}) = \ell(y, f_{\boldsymbol{\eta}_{\upper}}(z)).
\end{align}
Then define $G(\cdot)$ as the function below.
\begin{condition}
	\label{cond:trunc}
	The gradient of the loss function is bounded above by
	\begin{align}
	\sup_{\eta \in \Theta, x \in \mathcal{X}}
	\left\Vert
	\left .
	\nabla_{\boldsymbol{z}, \boldsymbol{\eta}_{\upper}}
	\tilde{\ell}(y, \boldsymbol{z}, \boldsymbol{\eta}_{\upper})
	\right |_{z = \psi(\boldsymbol{\theta}_1^\top \boldsymbol{x} + \boldsymbol{t})}
	\right\Vert _{\infty}
	\le G(y)
	\le C_0 (|y| + C_1) .
	\end{align}
	Here $G$ and the constants $C_0,C_1 > 0$ only depend on the number of layers and nodes in the upper layers. It does not depend on $p$.
\end{condition}
It is easy to show that the condition holds when $\ell$ is the mean squared error loss or the logistic loss.
Throughout the proofs, we will work with a general $G(y)$ until the end where we must specialize to the particular $G(y)$ in the classification/regression settings.

For the rest of this document, $c_i$ will denote some positive constant that can depend on $K$ and $X_{max}$.

To control the truncated empirical process, we begin with bounding the entropy of the function class
\[
\mathcal{G}_{r}=\left\{ \left\{ 
g_\eta(y,\boldsymbol x) \equiv \ell_{\eta}(y,\boldsymbol x)-\ell_{\eta^{*(\eta)}}(y,\boldsymbol x)\right\} 1\left\{ G(y)\le M_{n}\right\} :
\boldsymbol \eta\in\Theta_{r}\right\} 
\]
where $r > 0$ and
\begin{align}
\begin{split}
\Theta_{r} &= \left\{ \eta \in \Theta: 
\sum_{j =1}^p 
\Omega_\alpha(\boldsymbol\theta_{1,\cdot, j}-\boldsymbol\theta_{1,\cdot, j}^{*(\boldsymbol \eta)})+
\| \other(\boldsymbol \eta) - \other(\boldsymbol \eta^{*(\boldsymbol{\eta})})\|_{2} \le r
\right\}.
\label{eq:theta_r}
\end{split}
\end{align}
For a given set of $n$ observations $\{\boldsymbol{x}_i, y_i\}$ for $i = 1,...,n$, define the empirical norm 
\[
\|g\|_{\mathbb{P}_{n}}=\left(\frac{1}{n}\sum_{i=1}^{n}g^{2}(y_{i}, \boldsymbol x_{i})\right)^{1/2}
\]
The following lemma bounds the entropy of $\mathcal{G}_r$ with respect to $\|\cdot\|_{\mathbb{P}_{n}}$.
\begin{lemma}
	Suppose Condition~\ref{cond:trunc} holds and $EQ^*$ contains $Q$ equivalence classes.
	For any $r, M_{n}>0$, the following holds for all $u > 0$:
	\begin{align}
	H\left(u,\mathcal{G}_{r},\|\cdot\|_{\mathbb{P}_{n}}\right)
	& \le 
	c_0 \left[
	\log Q +
	m_{\other} \log\left(\frac{r M_n\sqrt{m_{\other}}  + u}{u}\right)
	\right . \\
	& \quad + 
	\left .
	m_1 \left(\frac{rX_{max} M_n\sqrt{m_1} }{c_{\alpha, m_1} u} \right )^{2}
	\log\left(1+ p \left(\frac{c_{\alpha, m_1} u}{rX_{max} M_n\sqrt{m_1} } \right )^{2} \right)
	\right ]
	\label{eq:entropy_bound}
	\end{align}
	where $c_{\alpha, m_1} = 1 - \alpha + \alpha/\sqrt{m_1}$.
	\label{lemma:entropy}
\end{lemma}

\begin{proof}
	To bound the entropy of $\mathcal{G}_r$, we partition $\mathcal{G}_r$ and bound the entropy of each partition.
	In particular, let $\Xi$ be the subset of $EQ^*$ such that no two $\boldsymbol{\eta}^*, \boldsymbol{\eta}^{*'} \in \Xi$ parameterize equivalent neural networks, e.g. $\Xi$ is composed of a representative member from each of the $Q$ equivalence classes.
	Then we partition $\mathcal{G}_r$ into
	\begin{align}
	\mathcal{G}_r \subseteq \cup_{\boldsymbol{\eta}^* \in \Xi} \mathcal{G}_{r, \boldsymbol{\eta}^*}
	\label{eq:entropy_partition}
	\end{align}
	where
	$\mathcal{G}_{r, \boldsymbol{\eta}^*} = \left\{ 
	g_{\boldsymbol{\eta}}
	:
	\sum_{j=1}^p \Omega_\alpha ({\boldsymbol\theta}_{1,\cdot,j} - \boldsymbol{\theta}_{1,\cdot,j}^*)
	+\|\other(\boldsymbol \eta)-\other(\boldsymbol \eta^*)\|_{2}
	\le r
	\right\}$.
	
	We can partition $\mathcal{G}_r$ according to \eqref{eq:entropy_partition} due to the following claim: 
	
	\underline{Claim}
	For any $\boldsymbol{\eta}^* \in EQ^*$, let
	\begin{align}
	\Theta_{r,EQ(\boldsymbol{\eta}^*)} &= \left\{\boldsymbol  \eta \in \Theta_r:\boldsymbol \eta^{*(\boldsymbol \eta)} \in EQ(\boldsymbol{\eta}^*) \right\}\\
	\Theta_{r,\boldsymbol{\eta}^*} &= \left\{\boldsymbol  \eta \in \Theta_r:
	\sum_{j=1}^p \Omega_\alpha ({\boldsymbol\theta}_{1,\cdot,j} - \boldsymbol{\theta}_{1,\cdot,j}^*)
	+\|\other(\boldsymbol \eta)-\other(\boldsymbol \eta^*)\|_{2}
	\le r
	\right\},
	\end{align}
	we have
	$$
	\left\{ 
	g_\eta(y,\boldsymbol x)
	:\boldsymbol \eta\in\Theta_{r,EQ(\boldsymbol{\eta}^*)}
	\right\}
	\subseteq
	\left\{ 
	g_\eta(y,\boldsymbol x)
	:\boldsymbol \eta\in\Theta_{r,\boldsymbol{\eta}^*}
	\right\}.
	$$
	\begin{proof}
		Consider any $\boldsymbol{\eta} \in \Theta_{r, EQ(\boldsymbol{\eta}^*)}$.
		Define $\pi$ as a function mapping $\boldsymbol{\eta}^{*(\boldsymbol \eta)}$ to $\boldsymbol{\eta}^*$ via permutation/sign-flip, i.e. $\pi(\boldsymbol{\eta}^{*(\eta)}) = \boldsymbol{\eta}^{*}$
		Letting $\tilde{\boldsymbol{\eta}} = \pi(\boldsymbol{\eta})$, we have $\pi(\boldsymbol{\eta}) - \boldsymbol{\eta}^* = \boldsymbol{\eta} - \boldsymbol{\eta}^{*(\boldsymbol \eta)}$ and so
		$$
		\sum_{j=1}^p \Omega_\alpha (\tilde {\boldsymbol\theta}_{1,\cdot,j} - \boldsymbol{\theta}_{1,\cdot,j}^*)
		+\|\other(\tilde{\boldsymbol \eta})-\other(\boldsymbol \eta^*)\|_{2}
		\le r.
		$$
		Thus every element in $\Theta_{r, EQ(\boldsymbol{\eta}^*)}$ can be mapped to an element in $\Theta_{r, \boldsymbol{\eta}^*}$ that is the same function $g_\eta$.
	\end{proof}
	
	Therefore let us bound $N\left(u,\mathcal{G}_{r, \boldsymbol{\eta}^*},\|\cdot\|_{\mathbb{P}_{n}}\right)$ for some $\boldsymbol{\eta}^* \in \Xi$.
	Consider any $\boldsymbol{\eta},\boldsymbol{\eta}^{'} \in\Theta_{r, \boldsymbol{\eta}^*}$.
	By the Mean Value Theorem and Condition~\ref{cond:trunc}, we have
	\begin{align*}
	|g_\eta(y,\boldsymbol x) - g_{\eta'}(y,\boldsymbol x)|
	& = \left|\ell_{\eta}(y,\boldsymbol x)-\ell_{\eta'}(y,\boldsymbol x)\right|1\left\{ G(y)\le M_{n}\right\}  \\
	& \le M_{n}\left(
	\left\|\left(\boldsymbol \theta-\boldsymbol \theta'\right)^{\top}\boldsymbol{x}\right\|_1
	+ \left \|\other(\boldsymbol \eta)- \other\left (\boldsymbol {\eta}^{'} \right ) \right \|_1
	\right)
	\end{align*}
	Squaring both sides and applying Cauchy Schwarz, we get 
	\begin{align}
	\begin{split}
	|g_\eta(y,\boldsymbol x) - g_{\eta'}(y,\boldsymbol x)|^2
	& \le  4 M_{n}^{2}\left(
	m_1 \left\|\left(\boldsymbol \theta-\boldsymbol \theta'\right)^{\top} \boldsymbol{x}\right\|^{2}_2
	+ m_{\other} \left \|\other(\boldsymbol \eta)- \other\left (\boldsymbol {\eta}^{'} \right ) \right \|_{2}^{2}
	\right)
	\label{eq:func_dist_bd}
	\end{split}
	\end{align}
	Therefore to bound the entropy of $\mathcal{G}_{r,\boldsymbol{\eta}^*}$, it suffices to bound the entropy of 
	$$
	\mathcal{J}_{r1} \coloneqq 
	\left \{
	\boldsymbol \beta \in \mathbb{R}^{m_{\other}} :
	\|\boldsymbol \beta-\other(\boldsymbol \eta^*)\|_{2}
	\le r
	\right \}
	$$
	and 
	$$
	\mathcal{J}_{r2} \coloneqq 
	\left\{
	\boldsymbol \theta^{\top} \boldsymbol{x}:
	\sum_{j=1}^p  \| \boldsymbol \theta_{1,\cdot,j} - \boldsymbol \theta_{1,\cdot,j}^* \|_1 \le \frac{r}{1 - \alpha + \alpha/\sqrt{m_1}}
	\right\}
	.
	$$
	
	The entropy of a ball with radius $r$ in $\mathbb{R}^{\other}$ is for all $u\ge0$,
	\begin{align*}
	H\left(u, \mathcal{J}_{r1} ,\|\cdot\|_{2}\right) 
	\le m_{\other} \log\left(\frac{4r + u}{u}\right).
	\end{align*}
	Applying a slight modification of the bound in Lemma 2.6.11 in \citet{wellner1996}, we get the following entropy bound (for all $u \geq 0$)
	\begin{align*}
	H\left(u, \mathcal{J}_{r2} ,\|\cdot\|_{P_n}\right)
	\le
	c_0 m_1 \left(\frac{rX_{max}}{(1 - \alpha + \alpha/\sqrt{m_1}) u} \right )^{2}
	\log\left(1+ p \left(\frac{(1 - \alpha + \alpha/\sqrt{m_1}) u}{rX_{max}} \right )^{2} \right).
	\end{align*}
	
	Putting these bounds together, the entropy of $\mathcal{G}_{r, \boldsymbol{\eta}^*}$ is bounded above for all $u\ge0$ by
	\begin{align*}
	H\left(u,\mathcal{G}_{r,\eta_{0}},\|\cdot\|_{\mathbb{P}_{n}}\right)
	& \le
	H\left(\frac{u}{2M_n\sqrt{m_{\other}}}, \mathcal{J}_{r1} ,\|\cdot\|_{2}\right) + H\left(\frac{u}{2M_n\sqrt{m_1}}, \mathcal{J}_{r2} ,\|\cdot\|_{P_n}\right)\\
	&\le m_{\other} \log\left(\frac{8 r M_n\sqrt{m_{\other}}  + u}{u}\right) \\
	& \quad + 
	c_0 m_1 \left(\frac{2 rX_{max} M_n\sqrt{m_1} }{(1 - \alpha + \alpha/\sqrt{m_1}) u} \right )^{2}
	\log\left(1+ p \left(\frac{(1 - \alpha + \alpha/\sqrt{m_1}) u}{2 rX_{max} M_n\sqrt{m_1} } \right )^{2} \right).
	\end{align*}
	
	Because we have $Q$ equivalence classes in $EQ^*$, we must add a $\operatorname{log} Q$ term, to attain our final bound.
	
\end{proof}	

Next we need to show that the symmetrized truncated empirical process term is small with high probability. We use the Rademacher random variables $W$, which are defined to have distribution $\Pr(W = 1) = \Pr(W = -1) =0.5$.

\begin{lemma}
	\label{lemma:sym}
	Assume the same conditions as Lemma~\ref{lemma:entropy}.
	Let $W_{1},...,W_n$ be $n$ independent Rademacher random variables.
	Let
	\begin{align}
	\delta =
	c_{3}
	M_n c_{\alpha, m, 2}
	\left(
	\sqrt{\log Q}
	+\frac{X_{max}}{c_{\alpha, m_{1}}} \sqrt{\log\left(\frac{p c_{\alpha, m, 2}^{2}}{m_{1}X_{\text{max}}^{2}}+1\right)}
	\log\left(
	n M_n c_{\alpha, m, 2} \right)
	\right)
	\label{eq:delta_def}
	\end{align}
	where $c_{\alpha, m, 2} = m_{1}+\sqrt{m_{\other}K} + \frac{X_{max}}{c_{\alpha, m_1}}$.
	Then for fixed $(\boldsymbol{x}_i, y_i)$ for $i = 1,...,n$, we have for all $r > 0$ and $T\ge 1$
	\begin{align*}
	&  \Pr\left(\sup_{\eta\in\Theta_{r}}
	\left|
	\frac{1}{n}\sum_{i=1}^{n}
	W_{i} 
	\left[
	\ell_{\eta}\left(y_{i}, \boldsymbol{x}_i \right) - \ell_{\eta^{*(\boldsymbol \eta)}}\left(y_{i}, \boldsymbol{x}_i \right)
	\right]\right|\ge
	T r \delta/\sqrt{n}
	\right)\\
	& \le c_{2}\exp\left(
	- c_3
	T^{2} 
	\left(
	\sqrt{\log Q}
	+\frac{X_{max}}{c_{\alpha, m_{1}}} \sqrt{\log\left(\frac{p c_{\alpha, m, 2}^{2}}{m_{1}X_{\text{max}}^{2}}+1\right)}
	\log\left(
	n M_n c_{\alpha, m, 2} \right)
	\right)^2
	(r^{2}\vee 1 )
	\right).
	\end{align*}
\end{lemma}
\begin{proof}
	We apply Lemma 3.2 in \citet{geer2000empirical}. First we check all the conditions are satisfied. 
	
	Using Condition~\ref{cond:trunc} and Taylor expansion, we can show that for any $\boldsymbol{\eta}, \boldsymbol{\eta}' \in \Theta_r$
	\begin{align}
	& \left|\ell_{\eta}(y, \boldsymbol{x})-\ell_{\eta'}(y, \boldsymbol{x})\right|\\
	& \le
	c_0 G(y)
	\left[
	\left(m_{1}+\sqrt{m_{\other} K} \right)
	\wedge
	\left(
	\frac{X_{max}}{c_{\alpha, m_1}} \sum_{j=1}^{p} \Omega_{\alpha}(\theta_{1,\cdot,j}-\theta_{1,\cdot,j}^{'})+\|\other(\boldsymbol{\eta})-\other(\boldsymbol{\eta}^{'})\|_{1}
	\right)\right]
	\end{align}
	for the same $c_{\alpha, m_1}$ defined in Lemma~\ref{lemma:entropy}.
	Thus we have
	\begin{align}
	&\sup_{\eta\in\Theta_{r}}
	\frac{1}{n}\sum_{i=1}^{n}
	\left|\ell_{\eta}(y_i, \boldsymbol{x}_i)-\ell_{\eta'}(y_i, \boldsymbol{x}_i)\right|^{2}1\left\{ G(y_{i})\le M_{n}\right\}\\
	& \le
	c_1 M_n^2
	\left(m_{1}+\sqrt{m_{\other}K} + \frac{X_{max}}{c_{\alpha, m_1}} \right)^2
	\left(
	r^2 \wedge 1
	\right)
	\label{eq:r_n_def}
	\end{align}
	Set $R_n^2$ equal to the right hand side of \eqref{eq:r_n_def}
	and let $\tilde{R}_{n}^2 =
	c_1 M_n^2
	\left(m_{1}+\sqrt{m_{\other}K} + \frac{X_{max}}{c_{\alpha, m_1}} \right)^2
	r^2$.
	Then we can bound Dudley's integral in Lemma 3.2 as follows
	\begin{align*}
	& \int_{r/n}^{R_{n}}H^{1/2}\left(u,\mathcal{G}_{r},\|\cdot\|_{n}\right)du\\
	& \le \int_{r/n}^{\tilde{R}_{n}}H^{1/2}\left(u,\mathcal{G}_{r},\|\cdot\|_{n}\right)du \\
	& \le
	c_{2}
	r M_n c_{\alpha, m, 2}
	\left(
	\sqrt{\log Q}
	+\frac{X_{max}}{c_{\alpha, m_{1}}} \sqrt{\log\left(\frac{p c_{\alpha, m, 2}^{2}}{m_{1}X_{\text{max}}^{2}}+1\right)}
	\log\left(
	n M_n c_{\alpha, m, 2} \right)
	\right)
	,
	\end{align*}
	where $c_{\alpha, m, 2} = m_{1}+\sqrt{m_{\other}K} + \frac{X_{max}}{c_{\alpha, m_1}}$.
	
	Let
	$\delta =
	c_{3}
	M_n c_{\alpha, m, 2}
	\left(
	\sqrt{\log Q}
	+\frac{X_{max}}{c_{\alpha, m_{1}}} \sqrt{\log\left(\frac{p c_{\alpha, m, 2}^{2}}{m_{1}X_{\text{max}}^{2}}+1\right)}
	\log\left(
	n M_n c_{\alpha, m, 2} \right)
	\right).
	$
	Then for any $T \ge 1$, we have
	\begin{align*}
	&  \Pr\left(\sup_{\eta\in\Theta_{r}}\left|\frac{1}{n}\sum_{i=1}^{n}W_{i}\left[\ell_{\eta}\left(y_i, \boldsymbol{x}_i\right)-\ell_{\eta^{*(\eta)}}\left(y_i, \boldsymbol{x}_i\right)\right]\right|\ge
	T r \delta/\sqrt{n}
	\right)\\
	& \le 
	c_4 \exp\left(- c_4 \frac{T^2 r^2 \delta^{2}}{R_{n}^{2}}\right)
	\\
	& = c_{4}\exp\left(
	- c_4
	T^{2} 
	\left(
	\sqrt{\log Q}
	+\frac{X_{max}}{c_{\alpha, m_{1}}} \sqrt{\log\left(\frac{p c_{\alpha, m, 2}^{2}}{m_{1}X_{\text{max}}^{2}}+1\right)}
	\log\left(
	n M_n c_{\alpha, m, 2} \right)
	\right)^2
	(r^{2}\vee 1 )
	\right).
	\end{align*}
\end{proof}

Using Lemma~\ref{lemma:sym} above, combined with a slight modification to symmetrization Corollary 3.4 in \citet{geer2000empirical}, one can bound the difference of the truncated empirical processes over random observations $(\boldsymbol{x}_i, y_i)$.
\begin{corollary}
	Suppose the same assumptions hold as in Lemma~\ref{lemma:entropy}.
	Let $\delta$ be defined as in \eqref{eq:delta_def}.
	Then for all $r > 0$ and $T\ge1$, the probability of that
	$$
	\sup_{\eta\in\Theta_{r}}\left|v_{n}^{\text{trunc}}(\boldsymbol \eta)-v_{n}^{\text{trunc}}(\boldsymbol \eta^{*(\boldsymbol \eta)})\right|\ge c_3 T r \delta/\sqrt{n}
	$$
	holds for random $(\boldsymbol{x}_i, y_i)$ for $i = 1,...,n$ holds with probability no greater than
	$$
	c_{4}\exp\left(
	- c_5
	T^{2} 
	\left(
	\sqrt{\log Q}
	+\frac{X_{max}}{c_{\alpha, m_{1}}} \sqrt{\log\left(\frac{p c_{\alpha, m, 2}^{2}}{m_{1}X_{\text{max}}^{2}}+1\right)}
	\log\left(
	n M_n c_{\alpha, m, 2} \right)
	\right)^2
	(r^{2}\vee 1 )
	\right).
	$$
\end{corollary}

Now we are ready to bound the scaled truncated empirical process over the entire parameter space $\Theta$.
\begin{lemma}
	\label{lemma:trunc_emp}
	Suppose the same assumptions as Lemma~\ref{lemma:entropy}.
	Let $\tilde{\lambda}= {c_{9} \delta}/{\sqrt{n}}$ for $\delta$ defined in \eqref{eq:delta_def}.
	Then for any $T\ge1$, we have
	\begin{align*}
	& \Pr\left(
	\sup_{\eta\in\Theta}
	\frac{\left|v_{n}^{\text{trunc}}(\boldsymbol{\eta})-v_{n}^{\text{trunc}}(\boldsymbol{\eta}^{*(\boldsymbol \eta)})\right|}{
		\left(\sum_{j=1}^p  \Omega_\alpha(\boldsymbol\theta_{1,\cdot,j}-\boldsymbol \theta_{1,\cdot,j}^{*(\boldsymbol \eta)})
		+\|\other(\boldsymbol \eta)- \other(\boldsymbol \eta^{*(\boldsymbol{\eta})})\|_{2}\right)\vee\tilde{\lambda}
	}
	\ge Tc_6\tilde{\lambda}\right)\\
	& \le c_7 \log n
	\exp\left(
	- c_{8}
	T^{2} 
	\left(
	\sqrt{\log Q}
	+\frac{X_{max}}{c_{\alpha, m_{1}}} \sqrt{\log\left(\frac{p c_{\alpha, m, 2}^{2}}{m_{1}X_{\text{max}}^{2}}+1\right)}
	\log\left(
	n M_n c_{\alpha, m, 2} \right)
	\right)^2
	\right)
	.
	\end{align*}
\end{lemma}

\begin{proof}
	We use a peeling argument by partitioning $\Theta$ into
	\[
	\Theta=\left[\cup_{j=J}^{\infty}\Theta_{j}\right]\cup\Theta_{J-1}
	\]
	where
	$
	J=\max\left\{ j:2^{j}\ge\tilde{\lambda}\right\} = - c_6 \log (n)$
	and
	\[
	\Theta_{j}=
	\left\{
	\boldsymbol \eta \in \Theta:
	2^{j-1}<
	\sum_{j=1}^p \Omega_\alpha ({\boldsymbol\theta}_{1,\cdot,j} - \boldsymbol{\theta}_{1,\cdot,j}^{*(\boldsymbol \eta)})+
	\|\other (\boldsymbol \eta)-\other(\boldsymbol \eta^{*(\boldsymbol \eta)})\|_{2}
	\le 2^{j}
	\right\}
	\quad \forall j=J,...,\infty
	\]
	and
	\[
	\Theta_{J-1}=\left\{
	\boldsymbol \eta\in \Theta:
	\sum_{j=1}^p \Omega_\alpha ({\boldsymbol\theta}_{1,\cdot,j} - \boldsymbol{\theta}_{1,\cdot,j}^{*(\boldsymbol \eta)})+
	\|\other (\boldsymbol \eta)-\other(\boldsymbol \eta^{*(\boldsymbol \eta)})\|_{2}
	\le2^{J-1}
	\right\} .
	\]
	Then using a peeling argument, we have
	\begin{align*}
	&  \Pr\left(
	\sup_{\eta\in\Theta}\frac{\left|v_{n}^{\text{trunc}}(\boldsymbol\eta)-v_{n}^{\text{trunc}}(\boldsymbol \eta_{0})\right|}{
		\left(
		\sum_{j=1}^p \Omega_\alpha ({\boldsymbol\theta}_{1,\cdot,j} - \boldsymbol{\theta}_{1,\cdot,j}^{*(\boldsymbol \eta)})+
		\|\other (\boldsymbol \eta)-\other(\boldsymbol \eta^{*(\boldsymbol \eta)})\|_{2}
		\right)
		\vee\tilde{\lambda}}
	\ge T \tilde{\lambda}
	\right)\\
	& \le \sum_{j=J-1}^{\infty} \Pr\left(\sup_{\eta\in\Theta_j}\frac{\left|v_{n}^{\text{trunc}}(\boldsymbol \eta)-v_{n}^{\text{trunc}}(\boldsymbol \eta_{0})\right|}{
		\left(
		\sum_{j=1}^p \Omega_\alpha ({\boldsymbol\theta}_{1,\cdot,j} - \boldsymbol{\theta}_{1,\cdot,j}^{*(\boldsymbol \eta)})+
		\|\other (\boldsymbol \eta)-\other(\boldsymbol \eta^{*(\boldsymbol \eta)})\|_{2}
		\right)
		\vee\tilde{\lambda}}\ge T \tilde{\lambda}
	\right)\\
	& \le c_{7}J
	\exp\left(
	- c_8
	T^{2} 
	\left(
	\sqrt{\log Q}
	+\frac{X_{max}}{c_{\alpha, m_{1}}} \sqrt{\log\left(\frac{p c_{\alpha, m, 2}^{2}}{m_{1}X_{\text{max}}^{2}}+1\right)}
	\log\left(
	n M_n c_{\alpha, m, 2} \right)
	\right)^2
	\right)
	\\
	&+\sum_{j=1}^{\infty}c_{9}
	\exp\left(
	- c_8
	2^{2j}
	T^{2} 
	\left(
	\sqrt{\log Q}
	+\frac{X_{max}}{c_{\alpha, m_{1}}} \sqrt{\log\left(\frac{p c_{\alpha, m, 2}^{2}}{m_{1}X_{\text{max}}^{2}}+1\right)}
	\log\left(
	n M_n c_{\alpha, m, 2} \right)
	\right)^2
	\right)
	\\
	& \le c_{10}
	\log n
	\exp\left(
	- c_{11}
	T^{2} 
	\left(
	\sqrt{\log Q}
	+\frac{X_{max}}{c_{\alpha, m_{1}}} \sqrt{\log\left(\frac{p c_{\alpha, m, 2}^{2}}{m_{1}X_{\text{max}}^{2}}+1\right)}
	\log\left(
	n M_n c_{\alpha, m, 2} \right)
	\right)^2
	\right).
	\end{align*}
	
\end{proof}

The last step to proving Theorem~\ref{thm:emp_proc} is to control the remainder term
\begin{align}
\label{eq:remain}
v_{n}^{\rem}(\boldsymbol\eta) =
v_{n}(\boldsymbol\eta)-v_{n}^{\text{trunc}}(\boldsymbol\eta).
\end{align}

To bound the remainder term, we incorporate the fact that $y - f^*$ is a sub-gaussian random variable.
\begin{lemma}
	\label{lemma:remainder}
	Suppose $\epsilon$ are independent sub-gaussian random variables.
	Suppose Condition~\ref{cond:trunc} holds.
	For any $\kappa \ge 1$, we can choose some $c > 0$ such that for $M_{n} > 0$, we have
	\begin{align}
	\label{eq:remain_bd}
	\begin{split}
	\Pr\left(
	\frac{
		\left|v_{n}^{\rem}(\boldsymbol \eta) - v_{n}^{\rem}(\boldsymbol{\eta}^{*(\boldsymbol \eta)})\right|
	}{
		\sum_{j=1}^p \Omega_\alpha ({\boldsymbol\theta}_{(j)} - \boldsymbol{\theta}_{0, (j)}^{(\boldsymbol \eta)})
		+\|\other(\boldsymbol \eta)-\other(\boldsymbol \eta^{*(\boldsymbol \eta)})\|_{2}
	}
	\ge \frac{c_{13}\delta}{\sqrt{n}}
	\right) \le O_{p}\left(
	\frac{c_{\alpha, m, 2} \exp\left(- M_{n}^{2} \right) }{\delta}
	\right)
	\end{split}
	\end{align}
	where $\delta$ was defined in \eqref{eq:delta_def}.
\end{lemma}

\begin{proof}
	To bound this probability, first we find an upper bound that has a tail behavior that is easier to bound.
	By Taylor expansion, we have that
	\begin{align}
	\label{eq:remainder}
	\begin{split}
	& 
	\frac{
		\left|\ell_{\eta}(y, \boldsymbol x)-\ell_{\eta^{*(\boldsymbol{\eta})}}(y, \boldsymbol x)\right|1\left\{ G(y)>M_{n}\right\}
	}{
		\sum_{j=1}^p \Omega_\alpha ({\boldsymbol\theta}_{(j)} - \boldsymbol{\theta}_{0, (j)}^{(\boldsymbol \eta)})
		+\|\other(\boldsymbol \eta)-\other(\boldsymbol \eta^{*(\boldsymbol \eta)})\|_{2}
	}
	\\
	& \le
	C_{0}
	\left(|y|+C_{1}\right)
	1\left\{ C_{0}\left(|y|+C_{1}\right)\right\}
	\left(\frac{X_{\text{max}}}{c_{\alpha,m_1}}+\sqrt{m_{\other}}\right)
	.
	\end{split}
	\end{align}
	Then we can upper bound \eqref{eq:remain_bd} by
	\begin{align}
	\Pr
	\left(
	C_{0}
	\left(\frac{X_{\text{max}}}{c_{\alpha,m_1}}+\sqrt{m_{\other}}\right)
	\left(\mathbb{P}_{n}-\mathbb{P}\right)\left[\left(|y_{i}|+C_{1}\right)1\left\{ C_{0}\left(|y_{i}|+C_{1}\right)\ge M_{n}\right\} \right]
	\ge\delta
	\right)
	\label{eq:easy_remain_bd}
	\end{align}
	The probability in \eqref{eq:easy_remain_bd} can be bounded using Markov's inequality and the fact that sub-gaussian random variables $Z$ satisfy
	\begin{align}
	E\left[\left|Z\right|1\left\{ \left|Z\right|\ge M\right\} \right]\le C'\exp\left(-c'M^{2}\right)
	\end{align}
	for constants $C',c'>0$ that only depend on the sub-gaussian parameters.
\end{proof}

Finally we prove Theorem~\ref{thm:emp_proc} by setting $M_n = O_p(\sqrt{\log n})$ and combining the results in Lemmas~\ref{lemma:trunc_emp} and \ref{lemma:remainder}.

\pagebreak
\bibliography{nnet_nonparam}

\end{document}